\documentclass[final,leqno,onefignum,onetabnum]{siamltex1213}
\usepackage{subcaption}
\usepackage[ruled,vlined]{algorithm2e}
\SetCommentSty{scriptsize}
\SetKwInput{Input}{Input}
\SetKwInput{Output}{return}
\usepackage{wrapfig}
\usepackage{amsmath, amssymb}
\usepackage{mathtools}
\newcommand{\shat}[1]{\vphantom{#1}\smash[t]{\hat{#1}}}
\usepackage{color}

\usepackage{booktabs}
\usepackage{siunitx}
\usepackage{xfrac}

\usepackage{tikz}
\usetikzlibrary{fit,calc}
\newcommand{\tikzmark}[1]{\tikz[overlay,remember picture] \node(#1) {};}

\newtheorem{remark}{Remark}[section]
\newtheorem{example}{Example}[section]

\newcommand\blfootnote[1]{  \begingroup
  \renewcommand\thefootnote{}\footnote{#1}  \addtocounter{footnote}{-1}  \endgroup
}

\graphicspath{{./figs/}}

\SetKwComment{tcc}{\{}{\}}
\SetKwInput{Input}{Input} \SetKwInput{Output}{return}
\SetKwIF{If}{ElseIf}{Else}{if}{}{else if}{else}{endif}
\SetKwFor{While}{while}{}{endw}
\SetKwFor{ForEach}{for each}{}{endfch}

\newcommand{\code}[1]{{\textnormal{\texttt{#1}}}}
\newcommand{\nnz}{\textnormal{\texttt{nnz}}}

\title{Reducing Parallel Communication in Algebraic Multigrid through
Sparsification}

\author{Amanda Bienz\footnotemark[2]
\and
Robert D. Falgout\footnotemark[3]
\and
William Gropp\footnotemark[4]
\and
Luke N. Olson\footnotemark[5]
\and
Jacob B. Schroder\footnotemark[6]
}

\begin{document}
\maketitle
\slugger{sisc}{xxxx}{xx}{x}{x--x} 
\renewcommand{\thefootnote}{\fnsymbol{footnote}}
\footnotetext[2]{\texttt{bienz2@illinois.edu},
                 \url{http://web.engr.illinois.edu/~bienz2/},
                 Department of Computer Science,
                 University of Illinois at Urbana-Champaign, Urbana, IL 61801}
\footnotetext[3]{\texttt{rfalgout@llnl.gov},
                 \url{http://people.llnl.gov/falgout2},
                 Center for Applied Scientific Computing,
                 Lawrence Livermore National Laboratory, Livermore, CA 94550}
\footnotetext[4]{\texttt{wgropp@illinois.edu},
                 \url{http://wgropp.cs.illinois.edu/},
                 Department of Computer Science,
                 University of Illinois at Urbana-Champaign, Urbana, IL 61801}
\footnotetext[5]{\texttt{lukeo@illinois.edu},
                 \url{http://lukeo.cs.illinois.edu},
                 Department of Computer Science,
                 University of Illinois at Urbana-Champaign, Urbana, IL 61801}
\footnotetext[6]{\texttt{schroder2@llnl.gov},
                 \url{http://people.llnl.gov/schroder2},
                 Center for Applied Scientific Computing,
                 Lawrence Livermore National Laboratory, Livermore, CA 94550}
\renewcommand{\thefootnote}{\arabic{footnote}}

\blfootnote{This material is based upon work supported by the National Science
Foundation Graduate Research Fellowship under Grant No. DGE-1144245. Any opinion,
findings, and conclusions or recommendations expressed in this material are those
of the authors and do not necessarily reflect the views of the National
Science Foundation. }

\blfootnote{This research is part of the Blue Waters sustained-petascale
computing project, which is supported by the National Science Foundation (awards
OCI-0725070 and ACI-1238993) and the state of Illinois. Blue Waters is a joint
effort of the University of Illinois at Urbana-Champaign and its National Center
for Supercomputing Applications.}

\blfootnote{This work was performed under the auspices of the U.S.
Department of Energy by Lawrence Livermore National Laboratory under Contract
DE-AC52-07NA27344 (LLNL-JRNL-673388)} 
\begin{abstract}
Algebraic multigrid (AMG) is an $\mathcal{O}(n)$ solution process
for many large sparse linear systems.  A hierarchy of progressively coarser
grids is constructed that utilize complementary relaxation and interpolation
operators.  High-energy error is reduced by relaxation, while low-energy error
is mapped to coarse-grids and reduced there.  However, large parallel
communication costs often limit parallel scalability.  As the multigrid
hierarchy is formed, each coarse matrix is formed through a triple matrix
product. The resulting coarse-grids often have significantly more nonzeros
per row than the original fine-grid operator, thereby generating
high parallel communication costs on coarse-levels. In this paper, we introduce
a method that systematically removes entries in coarse-grid matrices
after the hierarchy is formed, leading to an improved communication costs.
We sparsify by removing weakly connected or unimportant entries in the
matrix, leading to improved solve time.  The main trade-off is that if
the heuristic identifying unimportant entries is used too aggressively,
then AMG convergence can suffer.  To counteract this, the original
hierarchy is retained, allowing entries to be
reintroduced into the solver hierarchy if convergence is too slow.
This enables a balance between communication cost and convergence, as necessary.
In this paper we
present new algorithms for reducing communication and
present a number of computational experiments in support.
\end{abstract}

\begin{keywords}
multigrid, algebraic multigrid, non-Galerkin multigrid, high performance computing
\end{keywords}

\begin{AMS}\end{AMS}

\pagestyle{myheadings}
\thispagestyle{plain}
\markboth{Bienz, Falgout, Gropp, Olson, Schroder}{Sparsifying AMG}

\section{Introduction}

Algebraic multigrid (AMG)~\cite{McRu1982, BrMcRu1984, RuStu1987}
is an $\mathcal{O}(n)$ linear solver.  For standard discretizations
of elliptic differential equations, AMG is remarkably fast
\cite{hypre,AggCoarse, AggCoarse2}.
We consider AMG as a solver for the symmetric, positive definite
matrix problem
\begin{equation}
  A x = b,
\end{equation}
with $A \in \mathbb{R}^{n \times n}$ and $x, b \in \mathbb{R}^n$.
AMG consists of two phases, a setup and a solve phase.
The setup phase defines a sequence or \textit{hierarchy} of $\ell_{\text{max}}$ coarse-grid and
interpolation operators, $A_{1}, \dots, A_{\ell_{\text{max}}}$ and $P_{0},
\dots, P_{\ell_{\text{max}} - 1}$ respectively.  The solve phase iteratively
improves the solution through relaxation and coarse-grid correction.
The error not reduced by relaxation, called \textit{algebraically smooth},
is transferred to cheaper coarser-levels and reduced there.

The focus of this paper is on the communication complexity of AMG in a
distributed memory, parallel setting.  To be clear, we refer to the
\textit{communication complexity} as the time cost of interprocessor
communication, while referring to the \textit{computational complexity} as the
time cost of the floating point operations.  The \textit{complexity} or
\textit{total complexity} is then the cost of the algorithm, combining the
communication and computational complexities.

Both the convergence and complexity of an algebraic multigrid method are
controlled by the setup phase.  Highly accurate interpolation, which yields a
rapidly converging method, requires a slow rate of coarsening and often
denser coarse operators.  This large number of coarse-levels, as well as
increased density, correlates with an increase in the amount of work
required during a single iteration of the AMG solve phase.  In contrast, sparser
 interpolation and fast coarsening reduce the cost of a single iteration of AMG
 cycle, but often lead to a
deterioration in convergence~\cite{AggCoarse, AggCoarse2}.
Therefore, there is a trade-off between per-iteration complexity and
the resulting convergence factor.

The sparse matrices, $A_{1}, \dots, A_{\ell_{\text{max}}}$, in the multigrid
hierarchy are, by design, smaller in dimension, yet often decrease in sparsity.
As an example of this, Table~\ref{tab:density} shows the properties of a
 hierarchy for a 3D Poisson
problem with a 7-point finite difference stencil on a $100\times 100\times 100$
grid.  We see that as the problem size decreases on
coarse-levels, the average number of nonzero entries per row
increases.  Here we denote by $\nnz$ the number of nonzero entries in the
respective matrix.  Figure~\ref{figure:spy} depicts this effect for
this example, where we see that the density increases on lower levels in the
hierarchy.
\begin{table}[!ht]
\begin{center}
\begin{tabular}{r S[table-format=7.0] S[table-format=7.0] S[table-format=2.0] S[table-format=1.0]}
  {level}
& {matrix size}
& {nonzeros}
& {nonzeros per row} \\ {$\ell$} & {$n$} & {\nnz} & {\sfrac{\nnz}{$n$}} \\  \midrule
0 & 1000000 & 6940000 & 7  \\ 1 & 500000  & 9320600 & 19 \\ 2 & 83345   & 2775281 & 33 \\ 3 & 13265   & 745689  & 56 \\ 4 & 2207    & 208173  & 94 \\ 5 & 333     & 23843   & 72 \\ \end{tabular}
\end{center}
\caption{Matrix properties using classical AMG for a 3D Poisson problem.}\label{tab:density}
\end{table}
\medskip
\begin{figure}[!ht]
  \centering
	\includegraphics[width=\textwidth]{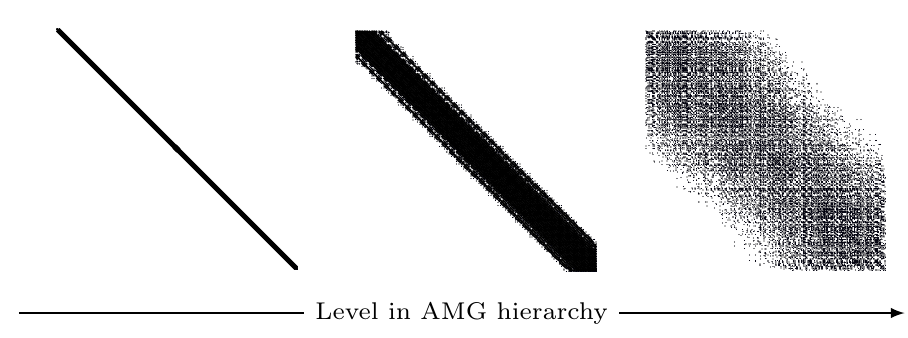}
	\caption{Matrix sparsity pattern using classical AMG for three levels in the hierarchy: $\ell = 0$, 3, 5. The full matrix properties are given in Table~\ref{tab:density}.}\label{figure:spy}
\end{figure}

In parallel, the increase in density (decrease in sparsity) on coarse-levels correlates
with an increase in parallel communication costs.
Figure~\ref{figure:level_times} shows this by plotting the time spent on each level
in an AMG hierarchy during the solve phase.  The time grows substantially on coarse-levels and
this is almost completely due to increased communication costs
from the decreasing sparsity. The time spent on smaller, coarse problems is
much larger than the time spent working on the original, finest-level
problem.  The test problem is again the 3D Poisson problem, using the hypre~\cite{hypre}
package with Falgout coarsening~\cite{falgout}, extended classical modified interpolation, and hybrid symmetric Gauss-Seidel relaxation.  This problem was run on 2048 processes with $10,000$ degrees-of-freedom per process.
\begin{figure}[!ht]
	\centering
	\includegraphics[width=.49\textwidth]{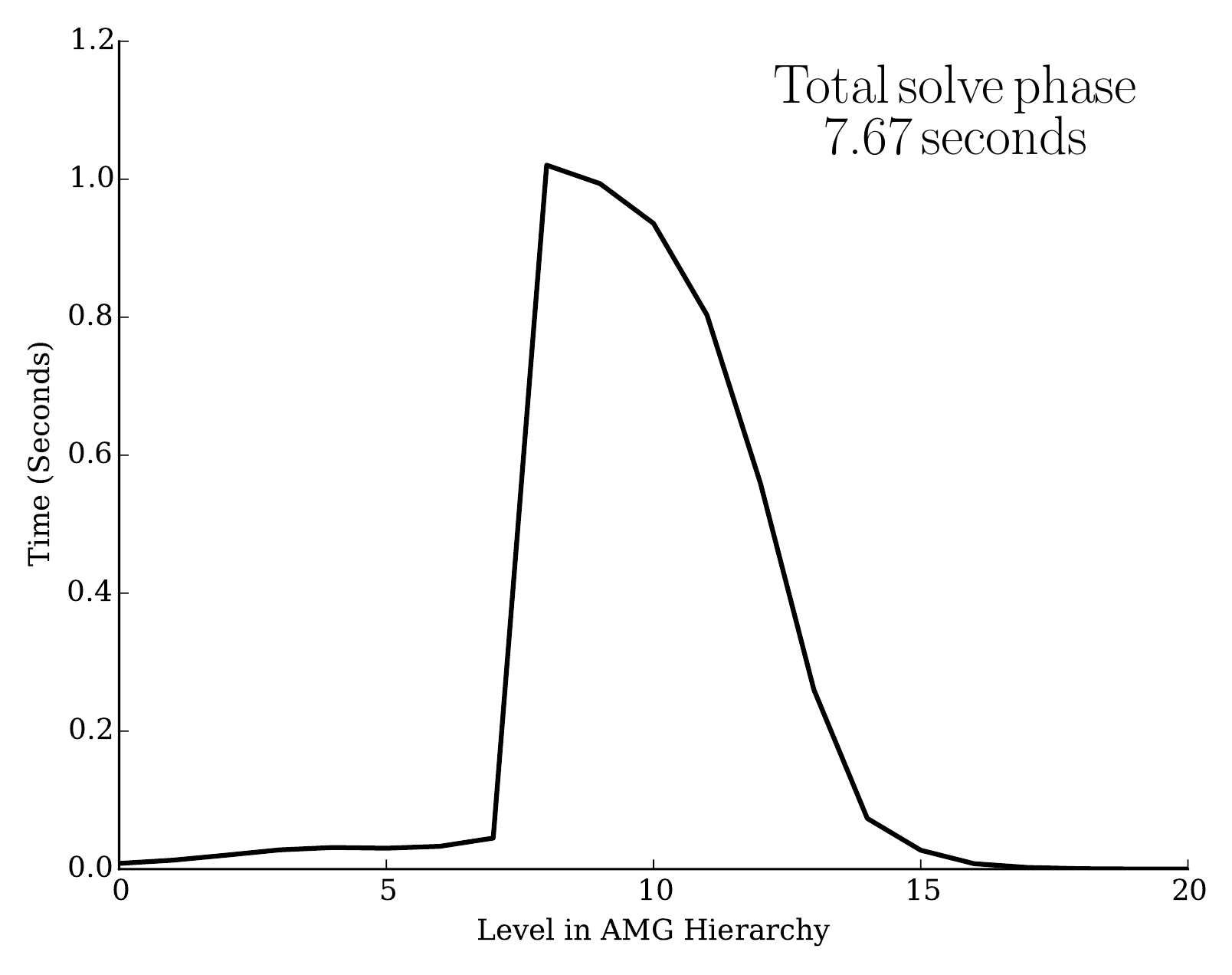}
	\includegraphics[width=.49\textwidth]{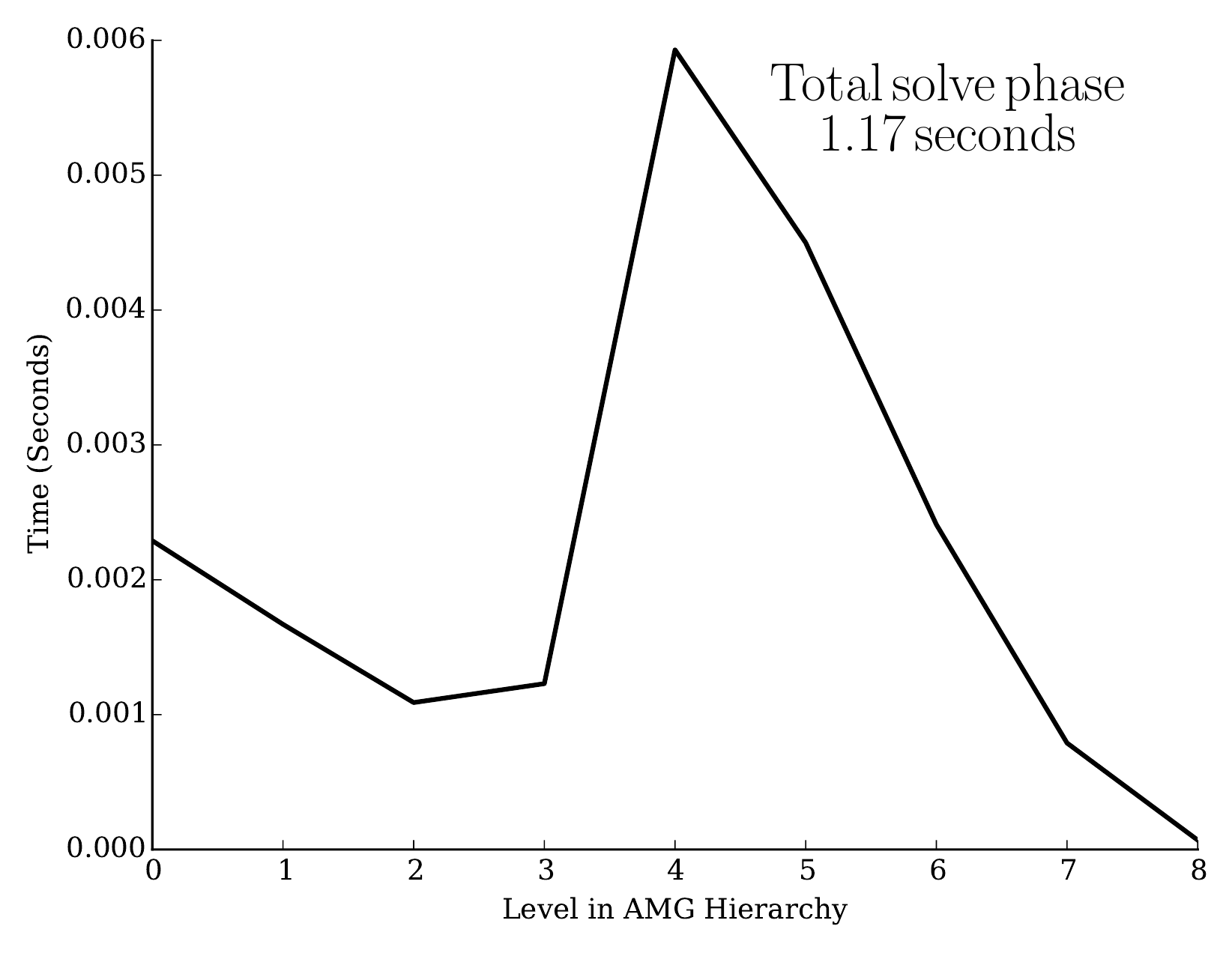}
	\caption{Left: Time spent on each level of the hierarchy during a single iteration of classical parallel AMG for a 3D Poisson
   problem.  Right: repeat experiment, but using aggressive HMIS coarsening.  The total time is much lower; however, the
   qualitative feature of expensive coarse-levels remains.}\label{figure:level_times}
\end{figure}

In this paper we introduce a method for controlling the communication complexity
in AMG\@.  The method increases the
sparsity of only coarse-grid operators ($A_\ell$,
$\ell=1,\dots,\ell_{\text{max}}$) by eliminating entries in $A_\ell$ with no effect on interpolation operators ($P_\ell$, $\ell=0, \ldots, \ell_{\text{max}-1}$).  This
 results in an improved balance
between convergence and per-iteration complexity in comparison to the
standard algorithm.  In addition, we develop an adaptive method which
allows nonzero entries to be reintroduced into the AMG hierarchy, should
the entry elimination heuristic be chosen too aggressively.  This allows for the method
to robustly maintain classic AMG convergence rates.

In the context of this paper, we define \textit{sparsity} and \textit{density} in
terms of the average number of nonzeros per row (or equivalently, the average
degree of the matrix).  In particular, density of a matrix $A_\ell$ of size $n_\ell$
is defined to be $\nnz(A_\ell)/n_\ell$.  The performance of AMG is closely
correlated with this metric, especially communication costs.  In addition, note
that if a matrix $A_\ell$ is ``sparser'' or ``denser'' under this definition, it is
also the case under the more traditional density metric, $\nnz(A_\ell)/n_\ell^2$.
Another advantage is that this measure yields a meaningful comparison between matrices
of different sizes.  For example, a goal of our algorithm is to generate coarse
matrices that have nearly the same sparsity as the fine grid matrix.

There are a number of existing approaches to reduce per-iteration communication complexity at
the cost of convergence.  Aggressive coarsening, such as HMIS~\cite{AggCoarse2} and
PMIS~\cite{AggCoarse2}, rapidly coarsen each level of the hierarchy.  These methods
greatly reduce the complexity of an algebraic multigrid iteration by reducing
both the number and the density of coarse operators.  While these coarsening
strategies reduce the cost of each iteration or cycle in the AMG solve phase,
they do so at the cost of accuracy, often leading to a reduction in
convergence.  Interpolation schemes such as \textit{distance-two
  interpolation}~\cite{DistTwo}, improve the convergence for aggressive coarsening,
but also result in an increase in complexity.

The result of aggressive coarsening and distance-two interpolation
does often lead to a notable reduction in the
time spent on each level in AMG, and a reduction in the total time to solution.
Figure~\ref{figure:level_times} shows that the time per level during the solve phase is reduced in comparison to standard coarsening, even though the same number of processes and problem size per core are used.  The use of HMIS coarsening is the only difference in problem settings between these two runs.  In the case of the simple 3D Poisson problem, there is only a nominal impact on convergence, yielding a large reduction in overall time spent in the solve
phase.  Nonetheless, while aggressive coarsening may reduce the total work
required during an iteration of AMG, the problem of expensive coarse-levels still persists.
For instance in~\cite{NonGal_Schroder}, it is noted that even when using these
best-practice parameters, parallel AMG for the 3D 7-point Poisson problem
produces coarse-grid operators with hundreds of nonzeros per row.  This can
be seen in Figure~\ref{figure:level_times}, where the time per-level
still spikes on coarse levels.

Another strategy for reducing communication complexity in AMG consists of systematically
adding sparsity into the interpolation operators~\cite{DistTwo}.  Removing
nonzeros from the interpolation operators reduces the complexity of the coarse-grid operators.
Modest levels of interpolation truncation are usually beneficial, however,
this process can also yield unpredictable impact on coarse-level performance if
used too aggressively.  Sparsity can alternatively be added into each coarse-grid operator by weighting the prolongation and restriction operators with entries of the appropriate Fourier series~\cite{Bolten2015}.

The typical approach to building coarse-grid operators, $A_{\ell}$, is to form
the Galerkin product with the interpolation operator: $A_{\ell + 1} = P_{\ell}^T
A_{\ell} P_{\ell}$.  This ensures a \textit{projection} in the coarse-grid
correction process and a guarantee on the reduction in error in each iteration
of the AMG solve phase.  On the other hand, the triple matrix product in the
Galerkin construction also leads to the growth in the number of nonzeros in coarse-grid matrices.  As
such, there are several approaches to constructing coarse operators that
do not use a Galerkin product and are termed \textit{non-Galerkin} methods.  These methods have been formed in a classical AMG setting~\cite{NonGal_Schroder} and also in a
smoothed aggregation~\cite{NonGal_Treister} context.  In general, these methods selectively remove
entries from coarse-grid operators, reducing the complexity of the multigrid
cycle.  Assuming the appropriate entries are removed from coarse-grid operators, the result is
a reduction in complexity with little impact on convergence.
However, if essential entries are removed, convergence deteriorates.

An alternative to limiting communication complexity is to directly determine the coarse-grid stencil,
an approach used in geometric multigrid.  For instance, simply
rediscretizing the PDE on a coarse-level results in the same stencil
pattern as for the original finest-grid operator, thus avoiding any increase in the number of nonzers in coarse-grid matrices.  More sophisticated approaches combine geometric and
algebraic information and include BoxMG~\cite{De1982,De1983} and PFMG~\cite{AsFa1996}, where a
stencil-based coarse-grid operator is built.  Additionally,
collocation coarse-grids (CCA)~\cite{CCA} have been used on coarse
levels to effectively limit the number nonzeros.  Yet, all these methods rely on geometric
properties of the problem being solved.  One exception is the extension of
collocation coarse-grids to algebraic multigrid (ACCA)~\cite{ACCA}, which has
shown similar performance to smoothed aggregation AMG\@.

The approach developed in this paper is to form non-Galerkin operator by
modifying \textit{existing} hierarchies.  The novel benefit of the proposed approach is
that it is applicable to most AMG methods, requires no geometric information, and
provides a mechanism for recovery if the dropping heuristic is chosen too
aggressively (see Section~\ref{section:adaptive}).
This paper is outlined as follows.   Section~\ref{section:amg} describes standard algebraic multigrid as
well as the method of non-Galerkin coarse-grids.
Section~\ref{section:sparse_methods} introduces two new methods for reducing the
communication complexity of AMG\@: Sparse Galerkin and Hybrid Galerkin.  Parallel performance
models for these methods are described in Section~\ref{section:perf}, and the
parallel results are displayed in Section~\ref{section:results}.  An adaptive
method for controlling the trade-off between communication complexity and convergence is
described in Section~\ref{section:adaptive}.  Finally,
Section~\ref{section:conclusion} makes concluding remarks.

\section{Algebraic Multigrid}\label{section:amg}

In this section we detail the AMG setup and solve phases, along with the basic
structure of a non-Galerkin method. We let the \textit{fine-grid} operator $A$ be
denoted with a subscript as $A_0$.

Algorithm~\ref{alg:amg_setup} describes the setup phase and begins with
\code{strength}, which identifies the strongly connected
edges\footnote{A degree-of-freedom $i$ is strongly connected to $j$ if algebraically
smooth error varies slowly between them.  Algebraically smooth error is not effectively
reduced by relaxation and has a small Rayleigh quotient~---~i.e., it's low in energy.
Strength information critically informs AMG how to coarsen and how to interpolate.
For more detail, see~\cite{BrMcRu1984, RuStu1987}.}
in the graph of
$A_{\ell}$ to construct the strength-of-connection matrix $S_{\ell}$.  From this,
$P_{\ell}$ is built in \code{interpolation} to interpolate vectors
from level $\ell+1$ to level $\ell$, with the goal to accurately interpolate
algebraically smooth functions.  For classical
AMG, \code{interpolation} first forms a disjoint splitting of the index set
$\{1,\dots,n\} = C \cup F$, where $C$ is a set of so-called coarse
degrees-of-freedom and where $F$ is a set of fine degrees-of-freedom.
The goal is to have algebraically smooth functions on $C$ accurately approximate
such functions on the full set $C \cup F$.  The size
of the coarse-grid is given by $n_{\ell+1} = |C|$, and an interpolation operator $P_{\ell}:
\mathbb{R}^{n_{\ell+1}}\rightarrow \mathbb{R}^{n_{\ell}}$ is constructed using $S_{\ell}$ and $A_{\ell}$
to compute sparse interpolation formulas accurate for algebraically smooth functions.
Finally, the coarse-grid operator is created through a Galerkin triple
matrix-product, $A_{\ell+1} = P_{\ell}^{T}A_{\ell}P_{\ell}$.  In a two-level setting,
this ensures the desirable property that the coarse-grid correction process
$(I - P_{\ell}^T A_{\ell} P_{\ell})A_{\ell}$ is an $A_{\ell}$-orthogonal projection on the error.  When
a non-Galerkin approximation is introduced, this property is lost.  Thus, the
most difficult task for us when designing a non-Galerkin algorithm is to
approximate the Galerkin product well.  If the approximation is poor, the
method can even diverge~\cite{NonGal_Schroder}.
\begin{algorithm}[!ht]
  \DontPrintSemicolon	\KwIn{    \begin{tabular}[t]{l l}
    $A_0$: & fine-grid operator\\
    \code{max\_size}:   & threshold for max size of coarsest problem\\
    $\gamma_{1}$, $\gamma_{2}$, $\dots$ & drop tolerances for each level\\
    \code{nongalerkin}: & (optional) non-Galerkin method
    \end{tabular}
  }
  \;
	\KwOut{    \begin{tabular}[t]{l}
    $A_{1}, \dots, A_{\ell_{\max}}$,\\
    $P_{0}, \dots, P_{\ell_{\max}-1}$
    \end{tabular}
  }
  \;
	\While{$\code{size}(A_\ell) > \code{max\_size}$}{		$S_{\ell}$ = \code{strength}{($A_{\ell}$)}\tcc*[r]{Strength-of-connection of edges}
		$P_{\ell}$ = \code{interpolation}{($A_{\ell}$, $S_{\ell}$)}\tcc*[r]{Construct interpolation and injection}
		$A_{\ell+1}$ = $P_{\ell}^{T}A_{\ell}P_{\ell}$\tcc*[r]{Galerkin product}
    \;
		\If(\tcc*[f]{(optional) described in Section~\ref{section:nongal}}){\code{nongalerkin}}{$A_{\ell+1} = \code{sparsify}(A_{\ell+1}$, $A_{\ell}$, $P_{\ell}, S_{\ell}$, $\gamma_{\ell}$)
\tcc*[r]{Remove nonzeros in $A_{\ell+1}$}
		}
	}
	\caption{\code{amg\_setup}}\label{alg:amg_setup}
\end{algorithm}

The density of each coarse-grid operator $A_{\ell + 1}$ depends on that of the
interpolation operator $P_{\ell}$. Even interpolation operators with modest
numbers of nonzeros typically lead to increasingly dense coarse-grid operators
\cite{NonGal_Schroder, Modeling}.  Algorithm~\ref{alg:amg_setup}
addresses this with the optional step \code{sparsify}, which triggers the
sparsification steps developed in this paper.  The non-Galerkin
approach~\cite{NonGal_Schroder} also fits within this framework.

The solve phase of AMG, described in Algorithm~\ref{alg:amg_solve} as a V-cycle,
iteratively improves an initial guess $x_{0}$ through use of the residual equation $A_0 e_0 = r_0$.
High energy error in the approximate solution is reduced through
relaxation in \code{relax}~---~e.g.\ Jacobi or Gauss-Seidel.  The remaining
error is reduced through coarse-grid correction: a combination of restricting the residual
equation to a coarser level, followed by interpolating and correcting with the resulting approximate error.
The coarsest-grid equation is computed with \code{solve}, using with a direct solution method.
\begin{algorithm}[!ht]
  \DontPrintSemicolon	\KwIn{    \begin{tabular}[t]{l}
    $x_0$, fine-level initial guess\\
    $b$, right-hand side\\
    $A_{1}, \dots, A_{\ell_{\max}}$\\
    $P_{0}, \dots, P_{\ell_{\max}-1}$
    \end{tabular}
  }
  \KwOut{    \begin{tabular}[t]{l}
    $x_0$, fine-level approximation
    \end{tabular}
  }
  \;
  \For{$i=0,\dots,\ell_{\text{max}}-1$}{    $\code{relax}(A_{\ell},x_{\ell},b_{\ell})$\tcc*[r]{Pre-smooth}
    $r_{\ell + 1} = P_{\ell}^T(b_{\ell} - A_{\ell}x_{\ell})$\tcc*[r]{Restrict residual}
  }
  $x_{\ell_{\text{max}}} = \code{solve}(A_{\ell_{\text{max}}}$, $r_{\ell_{\text{max}}})$\tcc*[r]{Coarsest-level direct solve}
  \For{$i=\ell_{\text{max}}-1,\dots,0$}{    $x_{\ell} = x_{\ell} + P_{\ell}x_{\ell + 1}$\tcc*[r]{interpolate and correct}
    $\code{relax}(A_{\ell},x_{\ell},b_{\ell})$\tcc*[r]{Post-smooth}
  }
	\caption{\code{amg\_solve}}\label{alg:amg_solve}
\end{algorithm}

The dominant computation kernel in Algorithm~\ref{alg:amg_solve} is the
sparse matrix-vector (SpMV) product, found in \code{relax} and
interpolation/restriction.  Typically relaxation dominates since
$A_{\ell}$ is larger and denser than $P_{\ell}$.  Thus, the performance on level $\ell$ of the
solve phase depends strongly on the performance of a single SpMV with $A_{\ell}$.

When performing parallel sparse matrix operations, a matrix $A$ is distributed across
processes in a row-wise partition, as shown in Figure~\ref{figure:processor}.
The local portion of the matrix is split into two groups: the diagonal
block, containing all columns of $A$ that correspond to local element of the vector;
and the off-diagonal block, corresponding to elements of the vector that are
stored on other processes.  For a SpMV, all off-process elements in the vector that correspond
to matrix nonzeros must be communicated.
Therefore, the density of a matrix
contribute to the cost of communication complexity in the SpMV operation.  This implies that the
decrease in sparsity on AMG coarse-levels leads to large communication costs and
often results in an inefficient solve phase~\cite{NonGal_Schroder, Modeling}.
\begin{figure}[!ht]
  \centering
	\includegraphics[width=.47\textwidth]{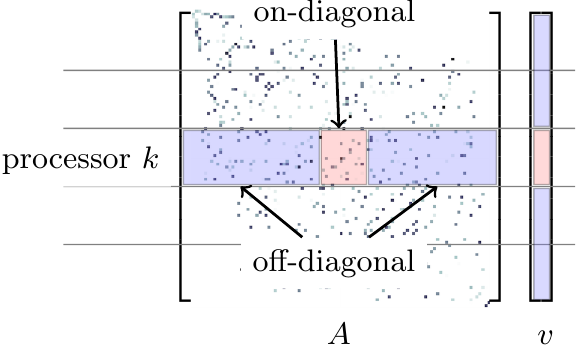}
\caption{Matrix $A$ and vector $v$ distributed across processes in a row-wise
partition.}\label{figure:processor}
\end{figure}

\subsection{Method of non-Galerkin coarse-grids}\label{section:nongal}

In this section we give an overview of the method~\cite{NonGal_Schroder} which
constructs hierarchies with coarse
operators that do not satisfy the Galerkin relationship where $A_{\ell+1} =
P_{\ell}^{T}A_{\ell}P_{\ell}$ for each level $\ell$.  The method of
\textit{non-Galerkin coarse-grids} forms the coarse-grid
operator through the Galerkin product, but then uses a sparsification step that generates
$\shat{A}_{\ell+1}$~---~see the call to \code{sparsify} in
Algorithm~\ref{alg:amg_setup}. As motivated in the previous section, fewer
nonzeros in the coarse-grid operator reduce the communication requirements.
The sparser matrix $\shat{A}_{\ell+1}$ replaces $A_{\ell+1}$ and is then used when
forming the remaining levels of the hierarchy, creating a dependency between
$\shat{A}_{\ell+1}$ and all successive levels as shown in
Figures~\ref{fig:depgal}~and~\ref{fig:depngal}.  Thus, this approach does not
preserve a coarse-grid correction corresponding to an $A$-orthogonal projection,
as described in Section~\ref{section:amg}.

In the following we use \code{edges}$(A)$, for a sparse matrix $A$, to represent
the set of edges in the graph of $A$.  That is, $\code{edges}(A) = \{(i,j)
\,\text{such that}\, A_{i,j} \ne 0\}$, where $A_{i,j} = {\left( A
\right)}_{i,j}$ is the ${(i,j)}^{\text{th}}$ entry of $A$.  In addition, we
denote $\shat{P}_{\ell}$ as the \textit{injection} interpolation operator that injects from
level $\ell+1$ to the $C$ points on level $\ell$ so that $\shat{P}_\ell$ is
defined as the identity over the coarse points.

The \code{sparsify} method for reducing the nonzeros in a matrix is described
in Algorithm~\ref{alg:sparsify}.  Here, the method selectively removes
small entries outside a minimal sparsity pattern\footnote{The goal of the
minimal sparsity pattern is to maintain, at the minimum, a stencil as wide for
the coarse-grid as exists for the fine-grid.  This is a critical heuristic for achieving
spectral equivalence between the sparsified operator and the Galerkin operator.
The current $\mathcal{M}$ achieves this in many cases.  It is possible in some
cases to reduce $\mathcal{M}$ further.  See~\cite{NonGal_Schroder} for more
details.} given by $\mathcal{M}_{\ell}$
where $\code{edges}(\mathcal{M}_{\ell}) =
\code{edges}(\shat{P}_{\ell}^{T}A_{\ell}P_{\ell} +
P_{\ell}^{T}A_{\ell}\shat{P}_{\ell})$.  For notational convenience, we set $A$
to $A_{\ell}$ as well as other operators in the algorithm.  For a given
tolerance $\gamma$, any entry $A_{i,j}$ with $(i,j) \notin \mathcal{M}$
and $\left\lvert A_{i,j} \right\rvert < \gamma \max_{k \neq i} \left\lvert
A_{i,k} \right\rvert$ is considered insignificant and is removed.  When entry
$A_{i,j}$ is removed, the value of $A_{i,j}$ is lumped to other entries that are
strongly connected to $A_{i,j}$, and $A_{i,j}$ is set to zero.  This reduces the
per-iteration communication complexity and heuristically targets spectral equivalence
between the sparsified operator and the Galerkin operator~\cite{NonGal_Schroder,
NonGal_Treister}.

There is a trade-off between the communication requirements and the convergence
rate.  Each entry in the matrix has a communication cost that is
dependent on the number of network links that the corresponding message travels in addition
to network contention. In addition, each entry in the matrix also
influences convergence of AMG, with large entries generally having larger impact (although this is not uniformly the case).
Any entry that has an associated communication cost
outweighing the impact on convergence should be removed. However, while it is
possible to
predict this communication cost based on network topology and message size, the
entry's contribution to convergence cannot be easily predetermined. When
dropping via
non-Galerkin coarse-grids, if the chosen drop tolerance is too large, too many
entries are removed and convergence deteriorates. Because the ideal drop
tolerance is problem dependent and cannot be
predetermined, it is likely that the chosen drop tolerance is suboptimal.
\begin{algorithm}[!th]
	\caption{\code{sparsify} from~\cite{NonGal_Schroder}}\label{alg:sparsify}
  \DontPrintSemicolon  \vspace{0.15cm}
	\KwIn{    \begin{tabular}[t]{l l}
      $A_{c}$ & coarse-grid operator\\
      $A$     & fine-grid operator\\
      $P$     & interpolation\\
      $\shat{P}$ & injection\\
      $S$ & classical strength matrix\\
      $\gamma$ & sparse dropping threshold parameter
    \end{tabular}
    }
    \vspace{0.15cm}
	\KwOut{$\shat{A}_{c}$, a sparsified $A_{c}$}
  \;
  $\shat{P} \leftarrow \code{form-injection()}$\;
  $\mathcal{M} = \code{edges}(\shat{P}^{T}AP + P^{T}A\shat{P})$\tcc*[r]{Edges in the minimal sparsity pattern}
	$\mathcal{N} = \emptyset$\tcc*[r]{Edges to keep in $A_c$}
	$\shat{A}_{c} = {\bf 0}$\tcc*[r]{Initialize sparsified $A_c$}
  \;
	\For{${\left( A_c \right)}_{i,j} \neq 0$}{    \If{$(i,j) \in \mathcal{M}$ {\code{or}} $\left\lvert {\left( A_c \right)}_{i,j} \right\rvert \geq \gamma \max_{k \neq i} \left\lvert {\left( A_c \right)}_{i,k} \right\rvert$}{           $\mathcal{N} \leftarrow \mathcal{N} \cup \{(i,j), (j,i)\}$\tcc*[f]{Add strong edges or the required pattern}
		}
	}
  \;
  \tikzmark{a}\For{${\left( A_c \right)}_{i,j} \neq 0$}{		\eIf{$(i,j) \in \mathcal{N}$}{          \vspace{0.15cm}
			${\left( \shat{A}_c \right)}_{i,j} = {\left( \vphantom{\shat{A}_c}A_c \right)}_{i,j}$\;
		}{			$\mathcal{W} = \{k \,|\, S_{j,k} \neq 0, (i,k) \in \mathcal{N}\}$\tcc*[r]{Find strong neighbors in the keep list}
			\For{$k \in \omega$}{        $\alpha = \frac{\left\lvert S_{j,k} \right\rvert}{\sum_{m \in \mathcal{W}}\left\lvert S_{j,m} \right\rvert}$\tcc*[r]{Relative strength to $k$}
				${\left( \shat{A}_c \right)}_{i,k} \leftarrow {\left(\shat{A}_c \right)}_{i,k} + \alpha {\left( A_c \right)}_{i,j}$\;
				${\left( \shat{A}_c \right)}_{k,i} \leftarrow {\left(\shat{A}_c \right)}_{k,i} + \alpha {\left( A_c \right)}_{i,j}$\;
				${\left( \shat{A}_c \right)}_{k,k} \leftarrow {\left(\shat{A}_c \right)}_{k,k} - \alpha {\left( A_c \right)}_{i,j}$\;
			}
		}
	}
\end{algorithm}
\begin{algorithm}[!th]
    \renewcommand{\thealgocf}{\arabic{algocf}b}
    \caption*{Diagonal Lumping -- Alternative {\bf for} loop (\textcolor{red}{\S~\ref{section:diaglump}})}      \DontPrintSemicolon    \tikzmark{b}\For{${\left( A_c \right)}_{i,j} \neq 0$}{        \nl\label{alg:ismax}$\code{ismax} \leftarrow \left\lvert {\left( A_c \right)}_{i,j} \right\rvert = \max_{k \neq i} \left\lvert {\left( A_c \right)}_{ik} \right\rvert \,\code{and}\, (i,k) \notin \mathcal{N} \,\forall\, k \neq i \,\code{and}\, \sum_{j}{A_{i,j}} = 0$\;
        \eIf(\tcc*[f]{Keep if entry is the single, maximum nonzero}){$(i,j) \in \mathcal{N}$ \code{or} \code{ismax}}{          \vspace{0.15cm}
        \nl${\left( \shat{A}_c \right)}_{i,j} = {\left( A_c \right)}_{i,j}$\label{alg:dropa}
      }(\tcc*[f]{Otherwise add to the diagonal}){${\left( \shat{A}_c \right)}_{i,i} \leftarrow {\left( \shat{A}_c \right)}_{i,i} +
  {\left( A_c \right)}_{i,j}$\;
\begin{tikzpicture}[remember picture,overlay]
\coordinate (aa) at ($(a)+(-0.5,-.5)$);
\coordinate (bb) at ($(b)+(-0.5,0.0)$);
\draw[->, red, thick] (a) -- (aa) -- (bb) -- (b);
\end{tikzpicture}
      }
    }
\end{algorithm}

Figures~\ref{figure:motivate_conv}~and~\ref{figure:motivate_num_sends} show the
convergence and communication complexity, respectively, of various AMG hierarchies for
solving a 3D Poisson problem with the method of non-Galerkin coarse-grids.
The original Galerkin hierarchy converges in the fewest number of iterations,
but has the highest communication complexity.  Non-Galerkin removes an ideal number of
nonzeros from coarse-grid operators (labeled \textit{ideal}) when no entries are
removed from the first coarse level, and all successive levels have a drop tolerance
of 1.0.  In this case, the communication complexity of the
solver is greatly reduced with little effect on convergence.  However, if the first coarse
level is also created with a drop tolerance of 1.0,
essential entries are removed (labeled \textit{too many}).  While the complexity of
the hierarchy is further reduced, but the method fails to converge.
\begin{figure}[!ht]
	\centering
	\begin{subfigure}[b]{0.49\textwidth}
		\centering
		\includegraphics[width=\textwidth]{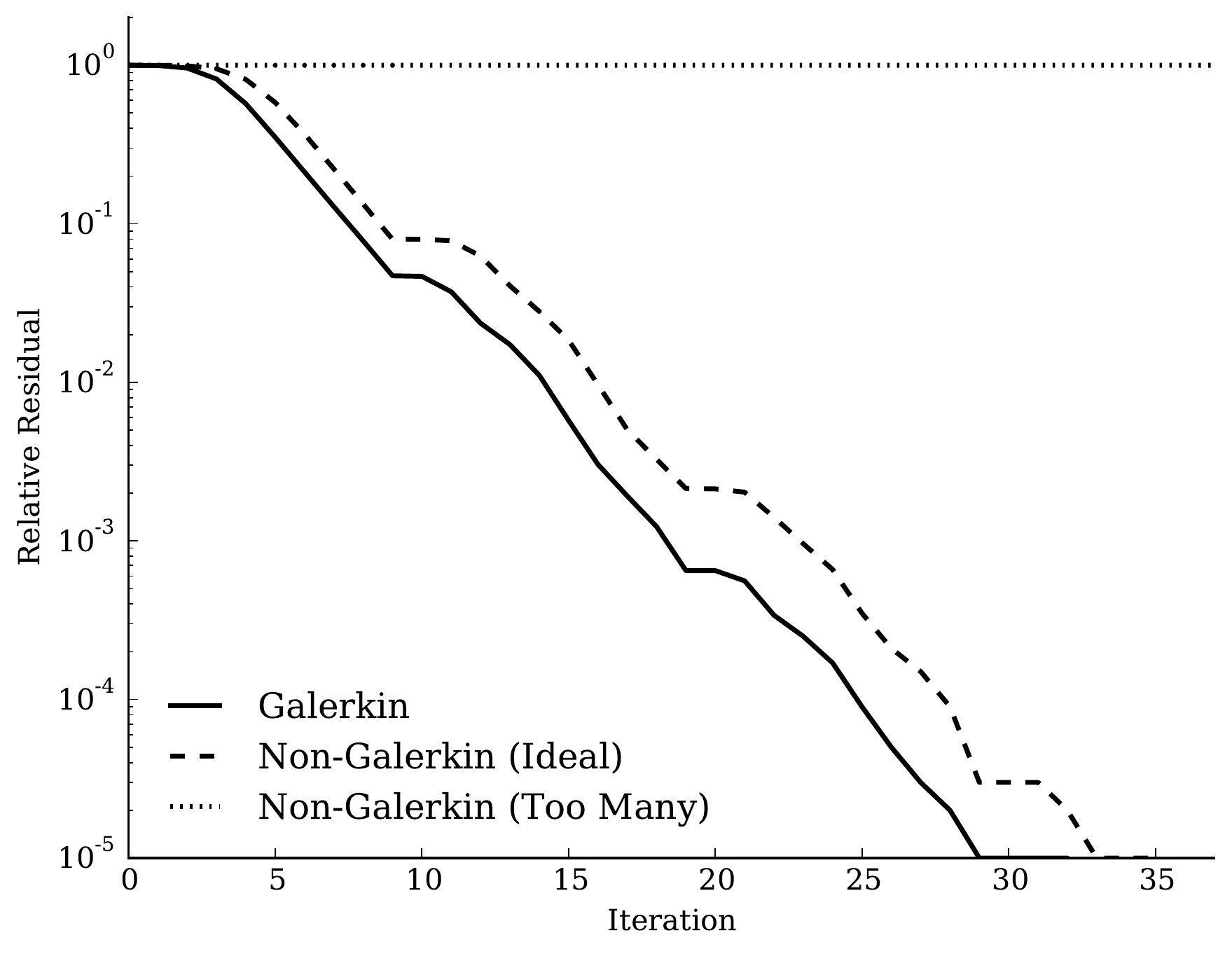}
		\caption{Relative residual per iteration}\label{figure:motivate_conv}
	\end{subfigure}
	\begin{subfigure}[b]{0.49\textwidth}
		\centering
		\includegraphics[width=\textwidth]{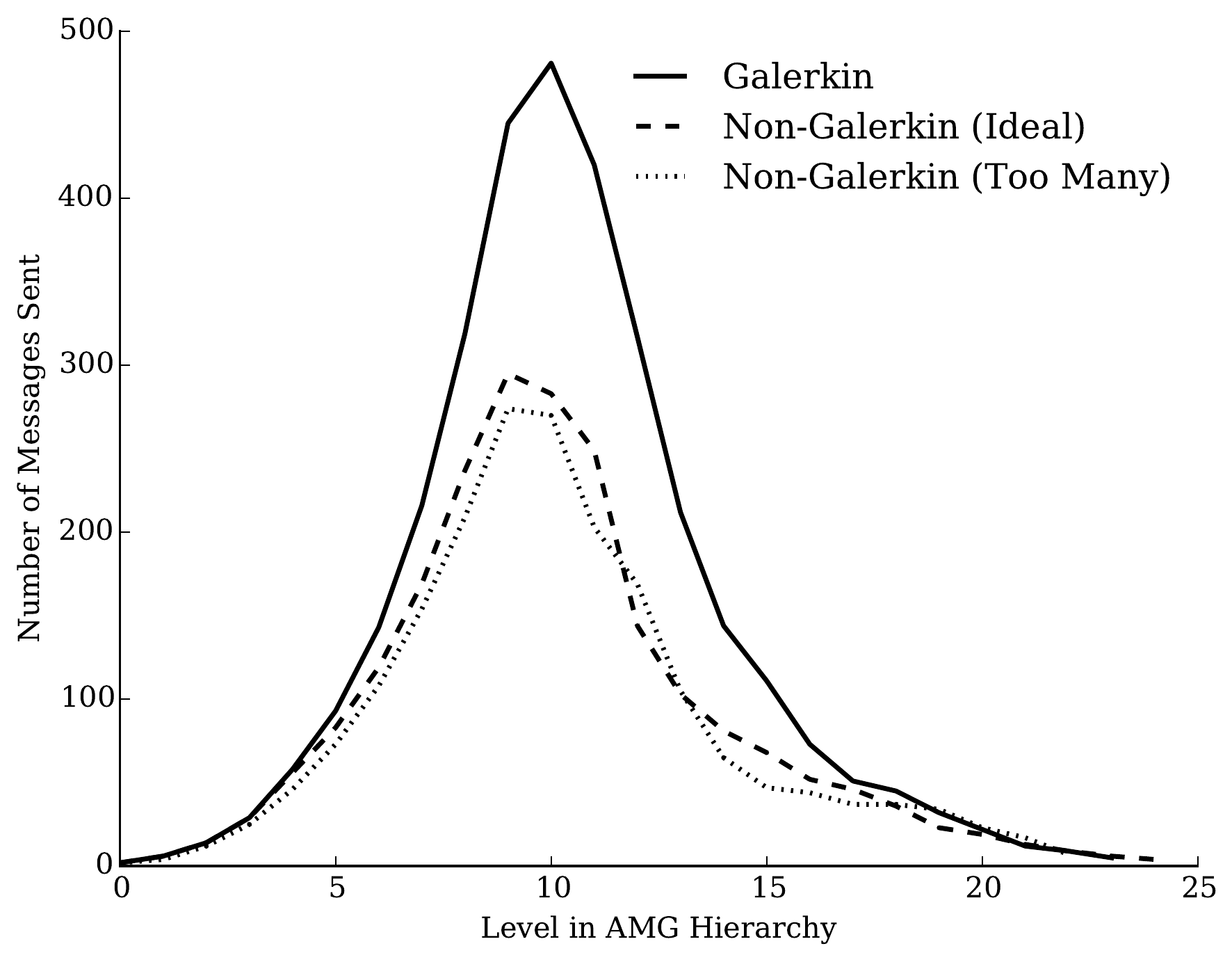}
		\caption{Number of MPI sends per level}\label{figure:motivate_num_sends}
	\end{subfigure}
	\caption{Adding sparsity to AMG hierarchy}
\end{figure}

If a large drop tolerance is chosen for non-Galerkin AMG, the effect on
convergence can be determined after one or two iterations of the solve phase.
At this point, if convergence is poor, eliminated entries can be re-introduced into the
matrix.  However, with this method, convergence improvements cannot be
guaranteed.  As shown in Algorithm~\ref{alg:amg_setup}, sparsifying on a
level affects all coarser-grid operators.  Hence, adding entries back into the
original operator does not influence the impact of their
removal on all coarser levels.
Figure~\ref{figure:nongalerkin_fixed} shows how re-adding entries
is ineffective by plotting the required communication costs verses the achieved
convergence for both
Galerkin and non-Galerkin AMG solve phases for the same 3D Poisson problem.
The data set \textit{Non-Galerkin (added back)} is generated by
removing entries with a drop tolerance of 1.0 (everything outside of
$\mathcal{M}$) on the first coarse-grid operator and 0.0 (retaining everything) on all
successive levels.  This results in a non-convergent method.  We then add
these removed entries back into
the first coarse-grid operator, but this does not reintroduce the entries which were removed from
coarser grid operators as a result of the non-Galerkin triple-matrix product
$P^{T}_{\ell}\shat{A_{\ell}}P_{\ell}$.  Figure~\ref{figure:nongalerkin_fixed}
shows that this hierarchy requires little
coarse-level communication after all entries have been reinstated to the first
coarse-grid operator.  However as the required entries are not added back into all coarser
grid operators, the method still fails to converge.\\
\begin{figure}[ht!]
	\centering
	\includegraphics[width=0.51\textwidth]{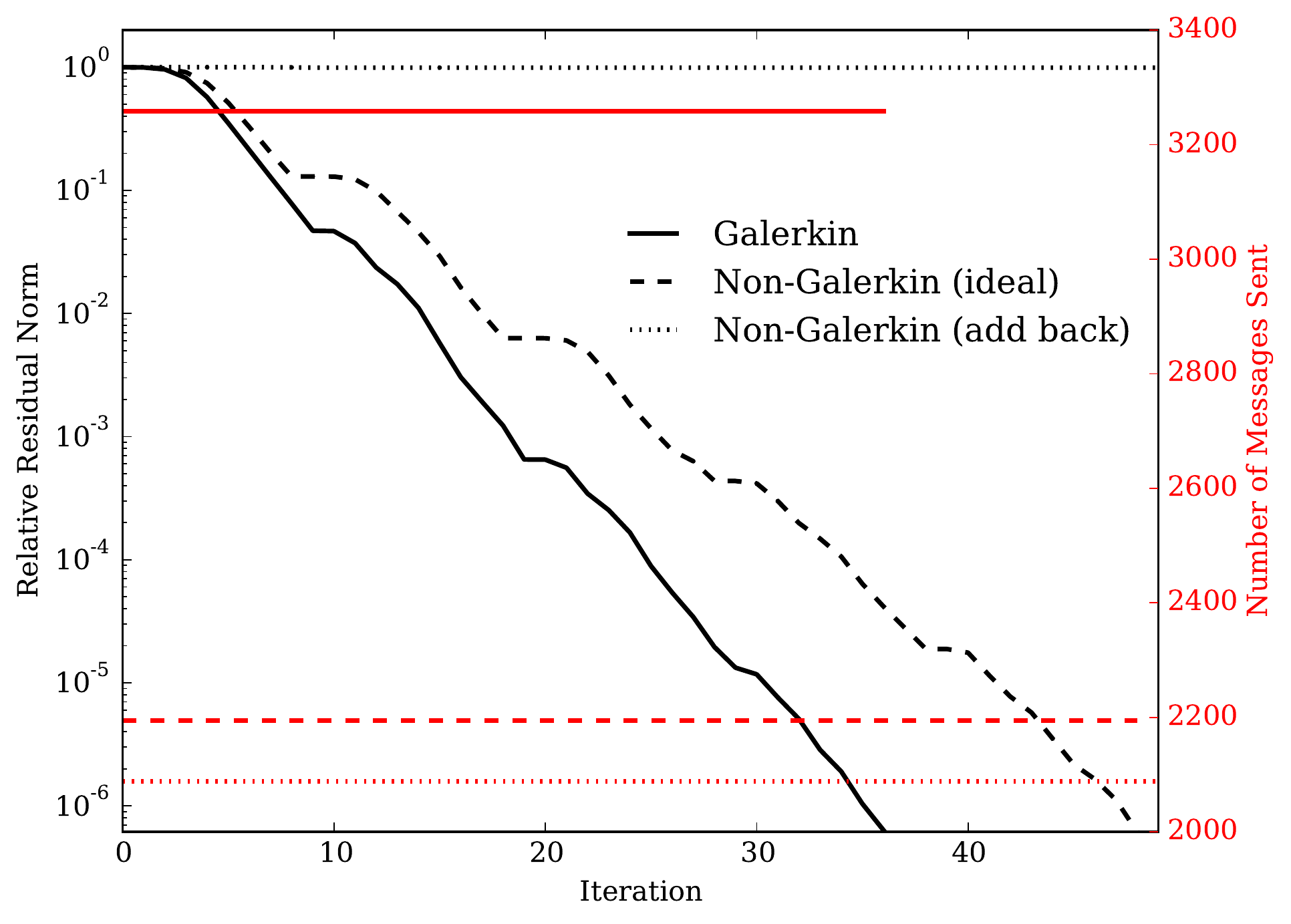}
	\includegraphics[width=0.47\textwidth]{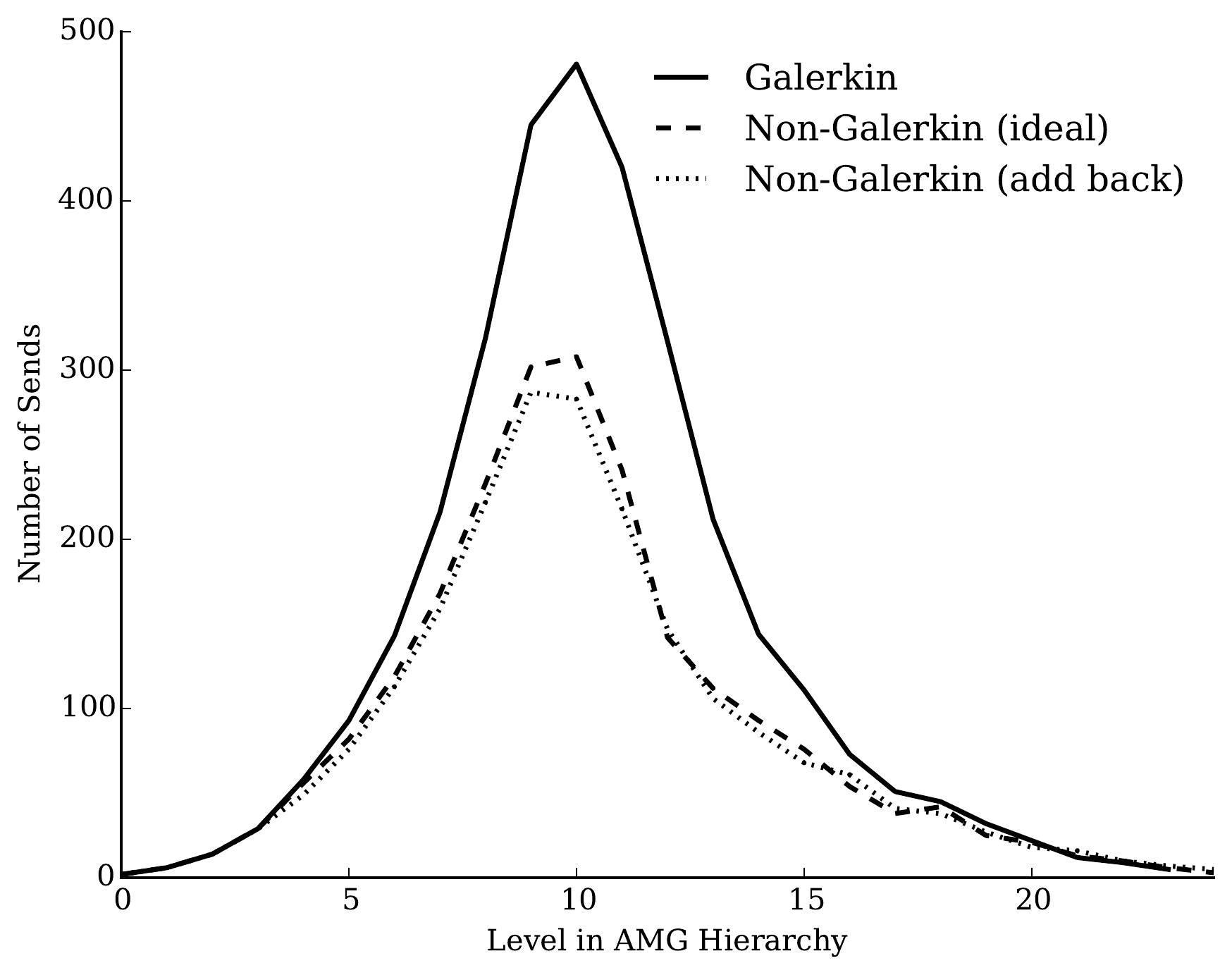}
	\caption{Convergence vs.\ communication of Galerkin and non-Galerkin
	hierarchies for the {\bf Poisson} problem.  Relative residual per AMG
	iteration (black) vs the number of MPI sends per iteration (red) (left),
	and number of sends per level in AMG hierarchy
	(right)}\label{figure:nongalerkin_fixed}
\end{figure}

\section{Sparse and Hybrid Galerkin approaches}\label{section:sparse_methods}

In this section we present two methods as alternatives to the method of
non-Galerkin coarse-grids.  The methods consist of forming the entire Galerkin hierarchy
before sparsifying each operator, yielding a lossless approach for increasing
sparsity in the AMG hierarchy.  The first method, which is called the Sparse
Galerkin method is described in Algorithm~\ref{alg:sparse_hybrid} (see
Line~\ref{alg:callsparse}).  Sparse Galerkin creates the entire Galerkin
hierarchy as usual.  The hierarchy is then thinned as a post-processing step to
remove relatively small entries outside of the minimal sparsity pattern
$\mathcal{M} = \shat{P}^{T} A P + P^{T}A\shat{P}$ using \code{sparsify}.
\begin{algorithm}[!ht]
	\caption{\code{sparse\_hybrid\_setup}}\label{alg:sparse_hybrid}
  \DontPrintSemicolon	\KwIn{    \begin{tabular}[t]{l l}
      $A_0$ & fine-grid operator\\
      \code{max\_size}:       & threshold for max size of coarsest problem\\
      $\gamma_{1}$, $\gamma_2$, $\dots$ & drop tolerances for each level\\
      \code{sparse\_galerkin} & Sparse Galerkin method\\
      \code{hybrid\_galerkin} & Hybrid Galerkin method\\
    \end{tabular}}
  \vspace{0.15cm}
	\KwOut{$\shat{A}_{1}, \dots, \shat{A}_{\ell_{\max}}$}
  \;
$A_{1}, \dots, A_{\ell_{\text{max}}}, P_{0}, \dots, P_{\ell_{\text{max}}-1}
 =\code{amg\_setup}(A_{0}, \code{max\_size}, \code{False})$\;   	$\shat{A}_{0} = A_{0}$\;
  	\For{$\ell \leftarrow 1$ \KwTo\ $\ell_{\text{max}}$}{  		\If{\code{sparse\_galerkin}}{        \nl\label{alg:callsparse}$\shat{A}_{\ell+1}$ = \code{sparsify}{($A_{\ell+1}$, $A_{\ell}$, $P_{\ell}$,
$S_{\ell}$, $\gamma_{\ell}$)}\tcc*[r]{Increase using the Sparse Method}
		}
		\ElseIf{\code{hybrid\_galerkin}}{      \nl\label{alg:callhybrid}$\shat{A}_{\ell+1}$ = \code{sparsify}{($A_{\ell+1}$, $\shat{A}_{\ell}$,
$P_{\ell}$, $S_{\ell}$, $\gamma_{\ell}$)}\tcc*[r]{Increase using the Hybrid Method}
		}
  	}
\end{algorithm}

The second method that we introduce is called Hybrid Galerkin since it combines
elements of Galerkin and Sparse Galerkin to create the final hierarchy.  The
method is again lossless, and is outlined in
Algorithm~\ref{alg:sparse_hybrid} (see Line~\ref{alg:callhybrid}).  After
the Galerkin hierarchy is formed, small entries outside are removed, this time
using a modified, minimal sparsity pattern of $\mathcal{M} =
\shat{P}^{T}\shat{A}P + P^{T}\shat{A}\shat{P}$.

The Sparse and Hybrid Galerkin methods retain the structure of the original
Galerkin hierarchy.  Consequently, these methods introduce error only into
relaxation and residual calculations.
The remaining components of each V-cycle in the solve
phase (see \code{amg\_solve}), such as restriction and interpolation are left
unmodified.  Therefore, the grid transfer operators do not depend on any
sparsification, as shown in Figure~\ref{figure:dependencies}.  Here, we see
that the Sparse Galerkin method does not use the modified (or sparsified)
operators to create the next coarse-grid operator in the hierarchy.
Conversely, Hybrid Galerkin uses the newly modified operator to compute the sparsity
pattern $\mathcal{M}$ for the next coarse-grid operator;
this process does not impact interpolation.
\begin{figure}[ht!]
	\begin{subfigure}[b]{0.24\textwidth}
		\centering
		\includegraphics[height=1.3in]{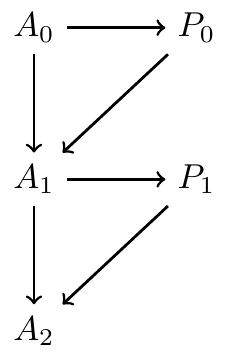}
		\caption{Galerkin}\label{fig:depgal}
	\end{subfigure}
	\begin{subfigure}[b]{0.24\textwidth}
		\centering
		\includegraphics[height=1.3in]{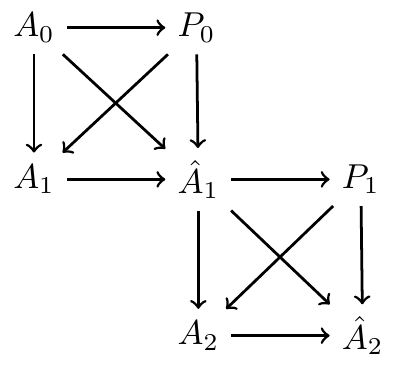}
		\caption{Non-Galerkin}\label{fig:depngal}
	\end{subfigure}
	\begin{subfigure}[b]{0.24\textwidth}
		\centering
\includegraphics[height=1.3in]{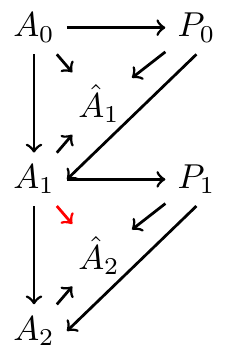}
\caption{Sparse Galerkin}\label{fig:depsp}
	\end{subfigure}
	\begin{subfigure}[b]{0.24\textwidth}
		\centering
\includegraphics[height=1.3in]{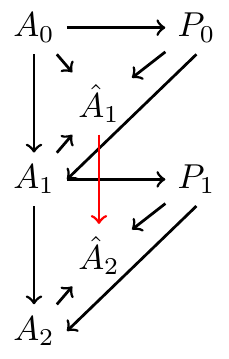}
\caption{Hybrid Galerkin}\label{fig:dephyb}
	\end{subfigure}
	\caption{Dependencies for forming each operator in the various AMG
	hierarchies.  The difference between Sparse and Hybrid Galerkin
	dependencies is highlighted in red.}\label{figure:dependencies}
\end{figure}

The new Sparse Galerkin and Hybrid Galerkin methods reduce the per-iteration
cost in the AMG solve cycle as less communication is required by each sparse,
coarse-grid operator.  However, high-energy error may also be relaxed at a
slower rate, yielding a reduction in the convergence factor. As a result, the
solve phase is more efficient when the reduction in communication outweighs the
change in convergence factor.

Similar to the method of non-Galerkin, it is difficult to predict
the impact of removing entries from $A_{c}$ on the relaxation
process.  However, as the structure of the Galerkin hierarchy is retained, the
convergence factor of the solve phase can be controlled on-the-fly. In our approach,
differences
between $A_{c}$ and $\shat{A_{c}}$ are stored while forming the sparse
approximations.
Subsequently, if the convergence factor falls below a
tolerance, entries can be reintroduced into the hierarchy,
allowing improvement of the convergence factor up to that of the original
Galerkin hierarchy~(see Section~\ref{section:adaptive}).

\subsection{Diagonal Lumping}\label{section:diaglump}

A significant amount of work is required in Algorithm~\ref{alg:sparsify}
to increase the sparsity of each coarse operator.  When forming non-Galerkin
coarse-grids, this additional setup cost is hidden by the reduced cost of the
sparsified triple matrix product $A_{c} = P^{T}\shat{A}P$.  However, as the
entire Galerkin hierarchy is initially formed as usual in our new methods
(Algorithm~\ref{alg:sparse_hybrid}) the additional work greatly reduces
the scalability of the setup phase, as shown in
Section~\ref{section:results_setup}.  This significant cost suggests using an
alternative method for sparsification of coarse-grid operators. When reducing the number of nonzeros from
coarse-grid operators with Sparse Galerkin or with
Hybrid Galerkin, the structure of the Galerkin hierarchy remains intact,
allowing a more flexible treatment of increasing sparsity in the matrix.  For
instance, one option is to remove entries by lumping to the diagonal rather
than strong neighbors, as described in
Algorithm~\ref{alg:sparsify}b.  This variation of \code{sparsify} is
beneficial for several reasons, including a much cheaper setup phase when compared to
Algorithm~\ref{alg:sparsify},
potential to reduce the cost of the solve phase, reduced storage
constraints for adaptive solve phases (see Algorithm~\ref{alg:adaptive}),
and retaining positive-definiteness of coarse operators.

Algorithm~\ref{alg:sparsify}b replaces the {\bf for} loop in
Algorithm~\ref{alg:sparsify}.  For each nonzero entry in the matrix, the
algorithm first checks if the entry is the maximum element in the row and if
all other entries in the row are selected for removal (see
Line~\ref{alg:ismax}).  In this case, the nonzero entry is not removed if there
is a zero row sum.

The method of diagonal lumping (Algorithm~\ref{alg:sparsify}b) results in a
significantly cheaper setup phase than Algorithm~\ref{alg:sparsify}.
The original non-Galerkin \code{sparsify} requires each
removed entry to be symmetrically lumped to significant neighbors.  As a result, the process of calculating
the associated strong connections requires a large amount of computation.
Furthermore, to maintain symmetry, all matrix entries that are not stored
locally must be updated, requiring a significant amount of interprocessor communication.
Lumping these entries to the diagonal eliminates both the computational and
communication complexities.

Eliminating the requirement of lumping to strong neighbors yields
potential for removing a larger number of entries from the hierarchy, further
reducing the communication costs of the solve phase.  The original version of
Algorithm~\ref{alg:sparsify} requires that an entry must have strong
neighbors to be removed, as its value is lumped to these neighbors.

While relaxing the restrictions of the original non-Galerkin \code{sparsify}
provides more opportunity to remove entries from the matrix, the diagonal
lumping also negatively influences convergence in some cases.  However, during
the solve phase, if convergence suffers, entries can be easily reintroduced into the
hierarchy, improving convergence.  As removed entries are only added to the
diagonal, the storage of both the sparse matrix along with removed entries is
minimal. In addition, these entries can be restored simply by adding their
values to the original positions, and subtracting these values from the
associated diagonal entries as shown in Algorithm~\ref{alg:adaptive}.  The
process of reintroducing these entries requires no interprocessor communication as well as
a low amount of local computational work.

Diagonal lumping also
preserves matrix properties such as symmetric positive-definiteness (SPD).  As
described in the following theorem, if the sparsity of a diagonally dominant,
SPD matrix
is increased using diagonal lumping, the resulting matrix remains
SPD\@. Consequently, Sparse and Hybrid Galerkin with diagonal lumping can be
used
in preconditioning many methods such as conjugate gradient.
It is important to note, that while SPD matrices are an attractive property for AMG, AMG methods do not guarantee diagonally dominance of the coarse-grid operators.  Yet, in many instances this property is preserved, for example for more standard elliptic operators.
\begin{theorem}\label{thm:spd}
Let $A$ be SPD and diagonally dominant.  If $\shat{A}$ is produced
by Algorithm~\ref{alg:sparsify}b, then it is
symmetric positive semi-definite and diagonally dominant.
\end{theorem}
\begin{proof}
Let $A$ be SPD with diagonal dominance,
\begin{equation}
\left|A_{i,i}\right| \geq \sum_{k\neq i}\left|A_{i,k}\right|, \forall i.
\end{equation}
Symmetry of $\shat{A}$ is guaranteed from the symmetry of both $A$ and the
$\mathcal{N}$ from Algorithm~\ref{alg:sparsify}.
For all off-diagonal entries $(i,j), (j,i) \in \mathcal{N}$,
\begin{equation}
  \shat{A}_{i,j}=A_{i,j}=A_{j,i}=\shat{A}_{j,i},
\end{equation}
by Line~\ref{alg:dropa} in Algorithm~\ref{alg:sparsify}b and the symmetry of $A$.

The positive-definiteness is guaranteed by the diagonal dominance and a
Gershgorin disc argument.  The proof proceeds by starting with the matrix $A$
and then considering the change made to $A$ by the elimination of each entry.
Initially, all the Gershgorin discs of $A$ are
strictly on the right-side of the origin, thus implying that all
eigenvalues are non-negative.  Then, assume that we eliminate some arbitrary
entry $A_{i,j}$, $(i,j) \in \mathcal{N}$.  This results in row $i$ being
updated
\begin{equation}
  A_{i,i} \leftarrow A_{i,i} + A_{i,j} \quad \mbox{and} \quad A_{i,j} \leftarrow 0
\end{equation}
If $A_{i,j} > 0$, then the center of the Gershgorin disc is shifted to the right, and the radius
shrinks, thus keeping the disc to the right of the origin and preserving definiteness.
If $A_{i,j} < 0$, then the center of the disc is
shifted to the left by $|A_{i,j}|$, but the radius of the disc also shrinks by $|A_{i,j}|$.
This also keeps the disc to the right of the origin and preserves semi-definiteness.
Furthermore, since each disc is never shifted to the left, diagonal dominance is also
preserved.
The proof then proceeds by considering all of the entries to be eliminated.
\end{proof}
\begin{remark}
   If any row of $A$ is strictly diagonally dominant, as often happens with Dirichlet boundary
   conditions, then $\shat{A}$ will be positive definite.  Essentially, Algorithm~\ref{alg:sparsify}b
   never shifts a Gershgorin disc to the left, so $\shat{A}$ can have no $0$ eigenvalue.
\end{remark}

\section{Parallel Performance}\label{section:perf}

In this section we model the parallel performance of Galerkin, non-Galerkin,
Sparse and Hybrid Galerkin using full lumping as in
Algorithm~\ref{alg:sparsify}, and Sparse and Hybrid Galerkin with diagonal
lumping as in Algorithm~\ref{alg:sparsify}b (labeled with \textit{Diag})
in order to illustrate the per-level costs associated with each method.  The
large increase in cost on coarse-levels in the Galerkin method (see
Figure~\ref{figure:level_times})
is due to the increase in coarse-level communication.

The solve phase of AMG (see Algorithm~\ref{alg:amg_solve}) is largely comprised
of sparse matrix-vector multiplication, thus we model each method by assessing
the cost of performing a SpMV on each level of the hierarchy.  We focus on the
operators $A_\ell$, as the work required for this matrix is more costly than
the restriction and interpolation operations. Specifically, we employ an
$\alpha$--$\beta$ model to capture the cost of the parallel SpMV based on the
number of nonzeros in $A$.  We denote $p$ as the number of processors,
$\alpha$ as the latency or startup cost of a message, and $\beta$ as the
reciprocal of the network bandwidth~\cite{Modeling, ModelSpMV}.  In addition, $\nnz_p$ represents the
average number of nonzeros local to a process, while $s_{p}$ and $n_{p}$ are
the maximum number of MPI sends and message size across all processors.
Finally, we use $c$ to represent the cost of a single floating-point operation.
With this we model the total time as
\begin{equation}\label{eq:model} T = 2 \, c \, \nnz_{p} + \max_{p} s_{p}(\alpha
  + \beta n_{p}).
\end{equation}
For the model parameters above we us the Blue Waters supercomputer at the
University of Illinois at Urbana-Champaign~\cite{BlueWaters,bw-in-vetter13}.
The latency and
bandwidth were measured through the HPCC benchmark~\cite{HPCC}, yielding
$\alpha= \num{1.8e-6}$ and $\beta = \num{1.8e-9}$. Since the achieved floprate depends
on matrix size, we determine the value of $c$ by timing the local SpMV\@.
Specifically, letting $\nnz_{\text{local}}$ be the number of nonzeros local to the
processor and $T_{\text{local}}$ the time to perform the local portion of the
SpMV, we compute $c = \sfrac{T_{\text{local}}}{2\nnz_{\text{local}}}$ for each
matrix in the hierarchy.

The minimal per-level cost associated with the non-Galerkin and Sparse/Hybrid
Galerkin methods occurs when entries are removed with a drop tolerance of
$\gamma=1.0$.  Using the model, (\ref{eq:model}), this is highlighted in
Figure~\ref{figure:perf_min} for both the Laplace and rotated anisotropic
diffusion problems (a full description of these problems is given in Section~\ref{section:results}.  We see that both non-Galerkin and Hybrid Galerkin have
potential to minimize the per-level cost.  However, when the per-level cost is
minimized, the convergence of AMG often suffers.  Therefore, less-aggressive
drop tolerances such as $\gamma<1.0$ may remove fewer entries, increasing the
per-level cost, but due to better convergence will improve the overall cost of the
solve phase.
\begin{figure}[ht!]
	\centering
  \includegraphics[width=0.49\textwidth]{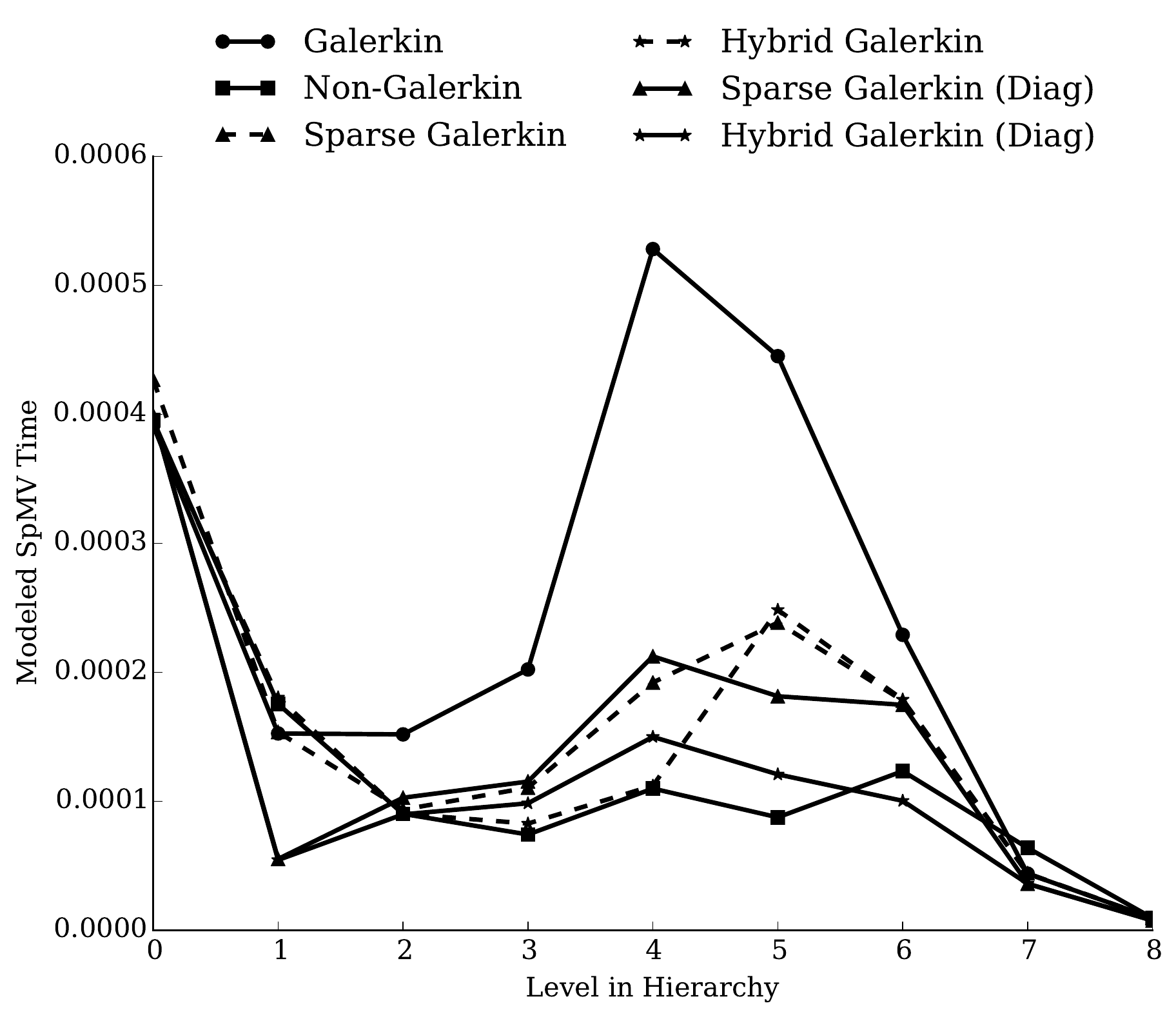}
  \hfill
  \includegraphics[width=0.49\textwidth]{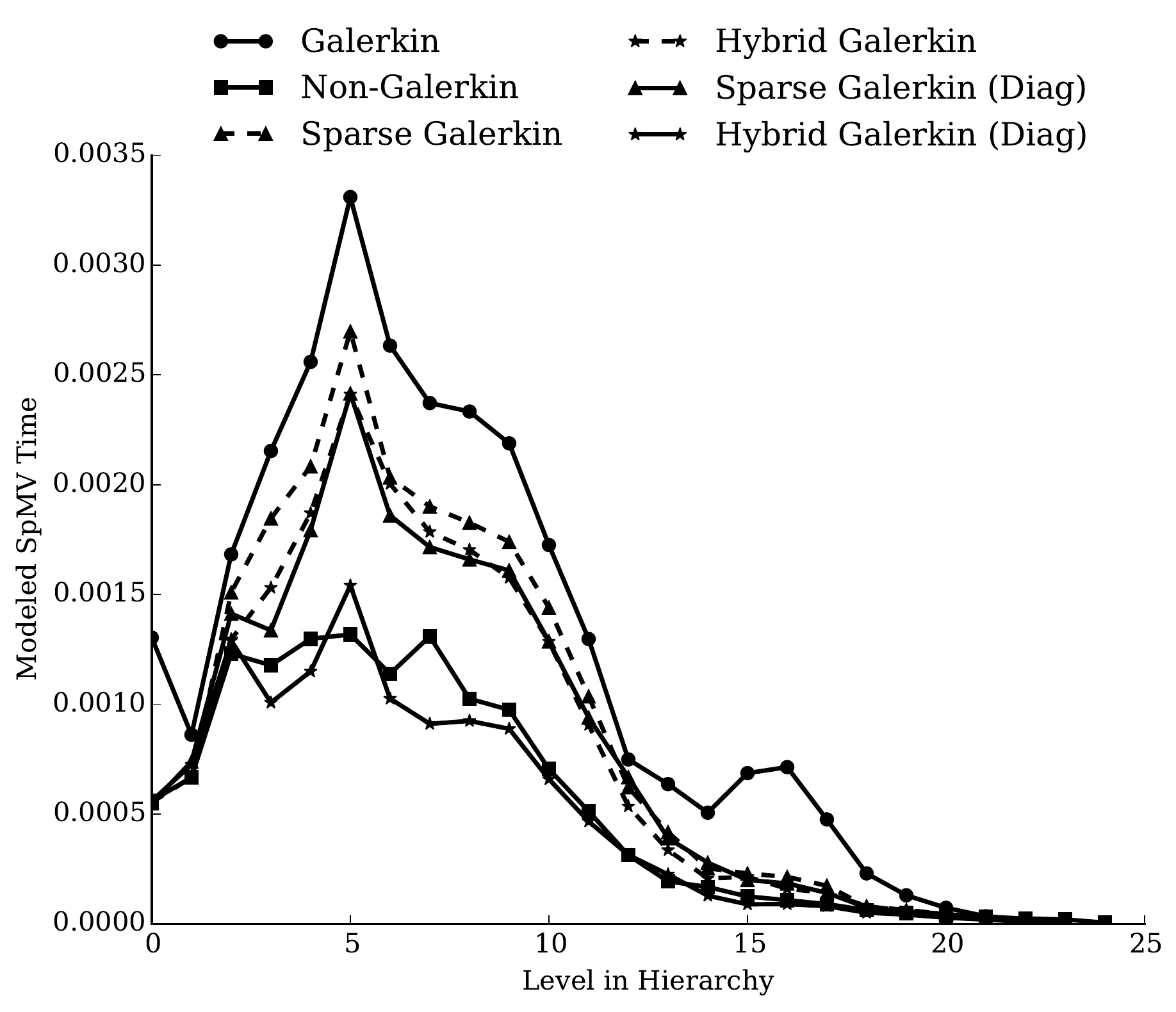}
	\caption{Modeled minimal cost of a single SpMV on each level of the AMG
	hierarchy for Laplace (left) and rotated anisotropic diffusion (right),
	for an aggressive drop tolerance of 1.0 on each level.}\label{figure:perf_min}
\end{figure}

Figure~\ref{figure:perf} models the cost of a SpMV on each level of the
hierarchy in the case of more realistic drop tolerances that are used to retain
convergence of the original method. These drop tolerances vary by level in the
AMG hierarchy, each containing a combination of $0.0$, $0.01$, $0.1$, and $1.0$.
For each test displayed in the model, six drop tolerance series were tested, and
we selected the smallest solve time.  The results underscore the simplicity of
the Laplacian, as removing entries can fortuitously improve convergence of this
problem.  For the rotated anisotropic diffusion problem the per-level cost of
non-Galerkin and Hybrid Galerkin increase in order to retain convergence.
However, the per-level cost of the non-Galerkin and Hybrid/Sparse Galerkin
methods are significantly decreased for levels near the middle of the hierarchy.
\begin{figure}[ht!]
	\centering
  \includegraphics[width=0.49\textwidth]{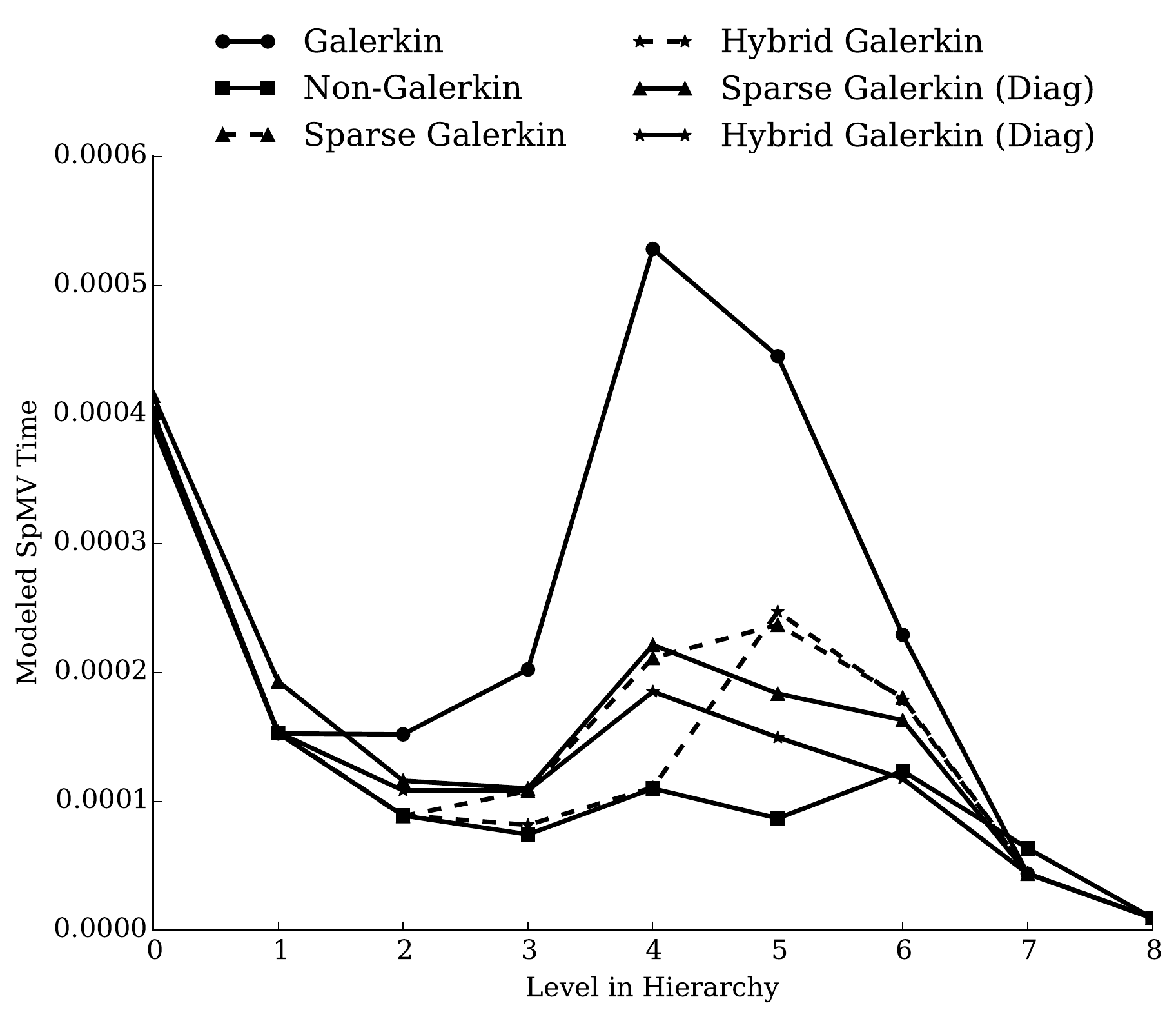}
  \hfill
  \includegraphics[width=0.49\textwidth]{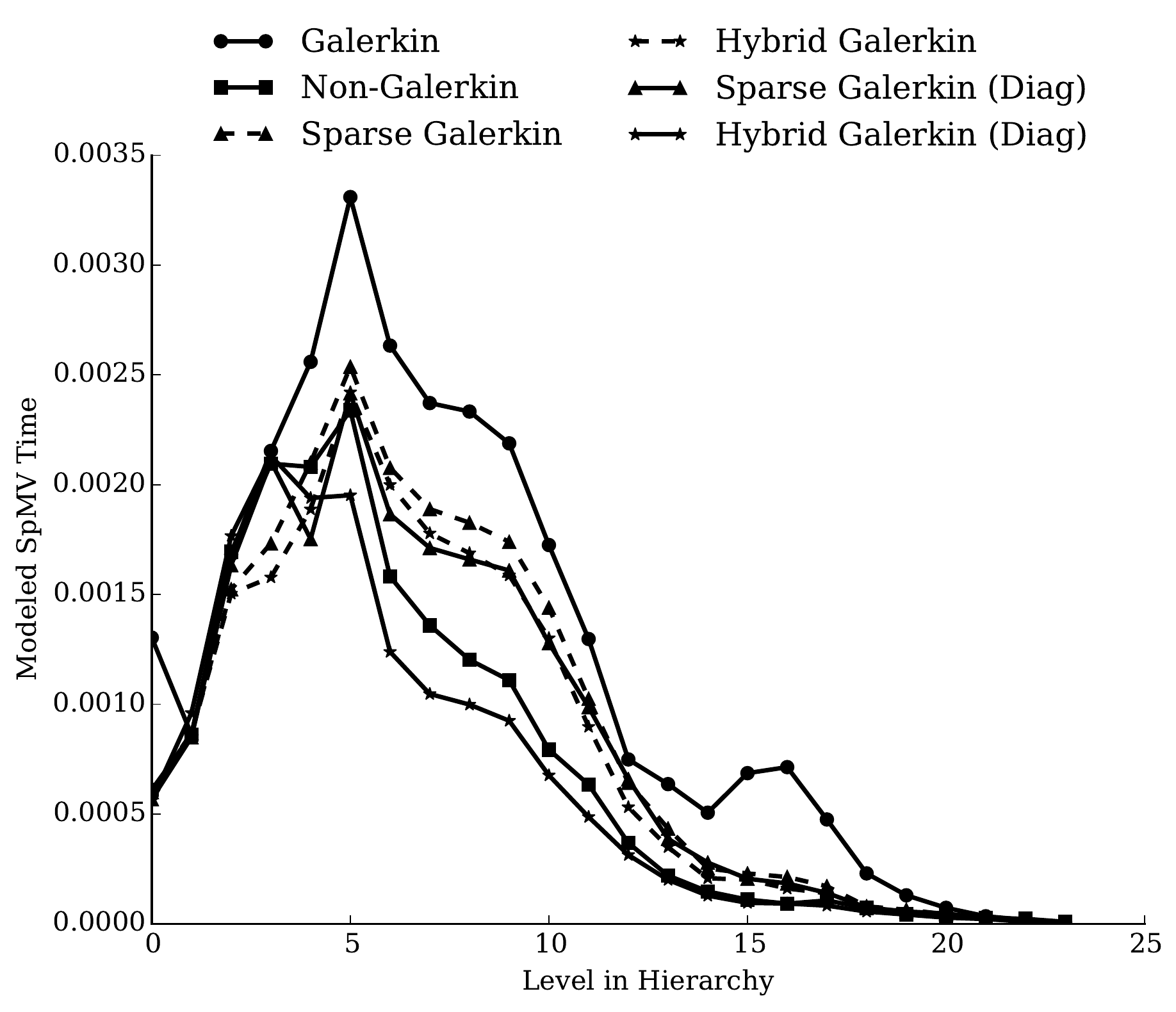}
\caption{Modeled minimal cost of a single SpMV on each level of the AMG
	hierarchy for Laplace (left) and rotated anisotropic diffusion (right)
	when a practical drop tolerance is used.}\label{figure:perf}
\end{figure}

\section{Parallel results for Sparse and Hybrid Galerkin}\label{section:results}

In this section we highlight the parallel performance of the Sparse and Hybrid
Galerkin methods.  We consider scaling tests on the familiar 3D Laplacian since
this is a common multigrid problem used to establish a baseline.  In order to
test problems where AMG convergence is suboptimal, we consider the 2D rotated
anisotropic diffusion problem.  Finally, we test our methods on a suite of
matrices from the Florida Sparse Matrix Collection.  All computations were
performed on the Blue Waters system at the University of Illinois at
Urbana-Champaign~\cite{BlueWaters}.  Each method was implemented and solved
with hypre~\cite{hypre, BoomerAMG}, using default parameters unless otherwise
specified. In summary, we compare the solve and setup times of the four methods
discussed in previous sections, preconditioning a Krylov method such as CG or
GMRES in each test:
\begin{description}
\item[Galerkin:] Classic coarsening in AMG, as outlined
                 in Algorithm~\ref{alg:amg_setup};
\item[non-Galerkin] The base algorithm presented in~\cite{NonGal_Schroder},
                    where $P^T A P$ is not used on coarse-levels;
\item[Sparse Galerkin] A new algorithm presented in
                       Algorithm~\ref{alg:sparse_hybrid}; and
\item[Hybrid Galerkin] A new algorithm presented in
                       Algorithm~\ref{alg:sparse_hybrid}.
\end{description}
In addition, we also consider the Sparse and Hybrid Galerkin methods with
diagonal lumping, as detailed in Algorithm~\ref{alg:sparsify}b.  The drop
tolerances for each method vary by level, using a combination of $0.0$, $0.01$,
$0.1$, and $1.0$ across the coarse-levels.  Six combinations of these drop
tolerances are tested for the various test cases, and the series yielding the
 minimum solve time for each is selected. \textit{note:} At $100,000$ cores, the
best drop tolerances from the second largest run size are used due to large costs
associated with running 6 drop tolerances at this core count.

We consider the diffusion problem
\begin{equation}
-\nabla \cdot K \nabla u = 0,
\end{equation}
with two particular test cases for our simulations:
\begin{description}
\item[3D Laplacian] Here, we use $K=I$ on the unit cube with homogeneous
  Dirichlet boundary conditions. Q1 finite elements are used to discretize the
  problem using a uniform mesh, leading to a familiar 27-point stencil.  The
  preconditioner formed for the 3D Laplacian uses aggressive coarsening (HMIS)
  and distance-two (extended classical modified) interpolation.  The
  interpolation operators were formed with a maximum of five elements per row,
  and hybrid symmetric Gauss-Seidel was the relaxation method.

  \item[Rotated, anisotropic diffusion]  In this case, we consider a diffusion
  tensor with homogeneous Dirichlet boundary conditions of the form $K=Q^{T}DQ$,
  where $Q$ is a rotation matrix and $D$ is a
  diagonal scaling defined as
  \begin{equation}
  Q = \left( \begin{matrix}
              \cos(\theta) & \sin(\theta)\\
              -\sin(\theta) & \cos(\theta)
           \end{matrix} \right)
         \qquad
D = \left (
\begin{matrix} 1 & 0\\ 0 & \epsilon \end{matrix} \right ).
       \end{equation}
        Q1 finite elements are used to discretize a uniform, square mesh. In
        the following tests we use $\theta = \frac{\pi}{8}$ and $\epsilon = 0.001$. In
        each case, the preconditioner uses Falgout
       coarsening~\cite{falgout}, extended classical modified interpolation and
       hybrid symmetric Gauss-Seidel.
\end{description}

Lastly, as problems with less structure result in increased density on coarse-levels, we
consider a subset from the Florida sparse matrix collection.
\begin{description}
       \item[Florida sparse matrix collection subset] We consider all real,
       symmetric, positive definite matrices from the Florida sparse matrix
       collection with size over 1,000,000
 degrees-of-freedom.  In addition we consider only the cases where GMRES
 preconditioned with Galerkin AMG converges in fewer than 100 iterations.
Each problem uses HMIS coarsening and so-called \textit{extended+i} interpolation if possible.  In
some cases, however, Galerkin AMG does not converge with these options; in these
cases Falgout coarsening and modified classical interpolation are used.
Relaxation for all systems is hybrid symmetric Gauss-Seidel.  \textit{note:}
When necessary for convergence, some hypre parameters, such as the minimum
coarse-grid size and strength tolerance, vary from the default.
\end{description}

The following results demonstrate that the diagonally lumped Sparse and Hybrid
Galerkin methods are able to perform comparably to non-Galerkin.  Non-Galerkin
and Sparse/Hybrid Galerkin all significantly reduce the per-iteration cost by
reducing communication on coarse-levels.  Since the method
of non-Galerkin is multiplicative in construction, the setup times are often
much lower in comparison to standard Galerkin.  However, Sparse and Hybrid do
not observe this benefit since the processing is \textit{post facto}. While the
per-iteration work is decreased for all methods, the convergence suffers for
the case of rotated anisotropic diffusion problems with non-Galerkin at large
scales. However, Sparse and Hybrid Galerkin converge at rates similar to the
original Galerkin hierarchy, yielding speedup in total solve times.

A strong scaling study shows that the anisotropic problem of a set size can be
most efficiently solved at larger scales when using non-Galerkin, but Hybrid
Galerkin performs comparably. Lastly, a strong scaling study of the subset of
Florida sparse matrix collection problems shows that non-Galerkin and Sparse/Hybrid
Galerkin each improve solve phase times for all matrices. While non-Galerkin
performs slightly better for many problems on $32$ cores, Hybrid Galerkin
outperforms the other methods for most problems at larger scales.

\subsection{Increasing sparsity in AMG Hierarchies}\label{section:results_hierarchy}

The significant number of nonzeros on coarse-levels creates large, relatively dense matrices near
the middle of the AMG hierarchy, yielding large
communication costs for each SpMV performed on these levels. As the solve
phase of AMG consists of many SpMVs on each level of the hierarchy, the time
spent on coarse-levels can increase dramatically.  Sparse, Hybrid, and non-Galerkin can all reduce both the cost associated with communication as well as the time spent on each level during a solve phase.

Figure~\ref{figure:level_times_512_large} shows the time spent on each level of
the hierarchy during a single iteration of AMG, for both test cases with 10,000
degrees-of-freedom per core using 8192 cores.  Both the method of non-Galerkin
coarse-grids, as well as the Sparse and Hybrid Galerkin methods, reduce
the time required on levels near the middle of the hierarchy.  Non-Galerkin more greatly reduces the time spent on middle levels of the hierarchy for the Laplace problem than Sparse and Hybrid Galerkin.  However, for the anisotropic problem, diagonally-lumped Hybrid Galerkin reduces level-time equivalently.  This is due to a
large reduction in the number of messages required in each SpMV as shown in
Figure~\ref{figure:comm_512_large}. The reduction in total size of all
messages communicated is relatively small.
\begin{figure}[ht!]
	\centering
\includegraphics[width=0.49\textwidth]{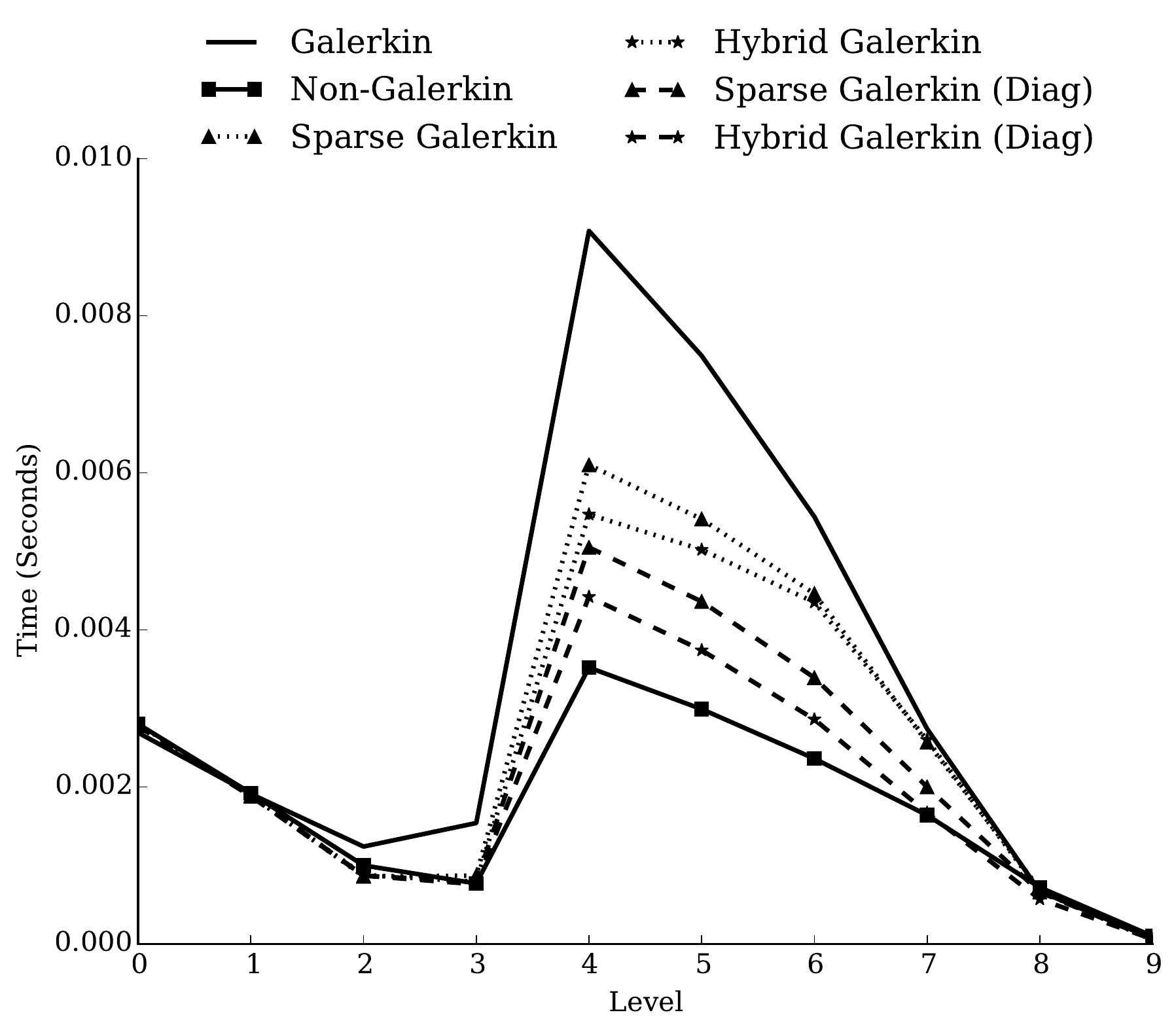}
\hfill
\includegraphics[width=0.49\textwidth]{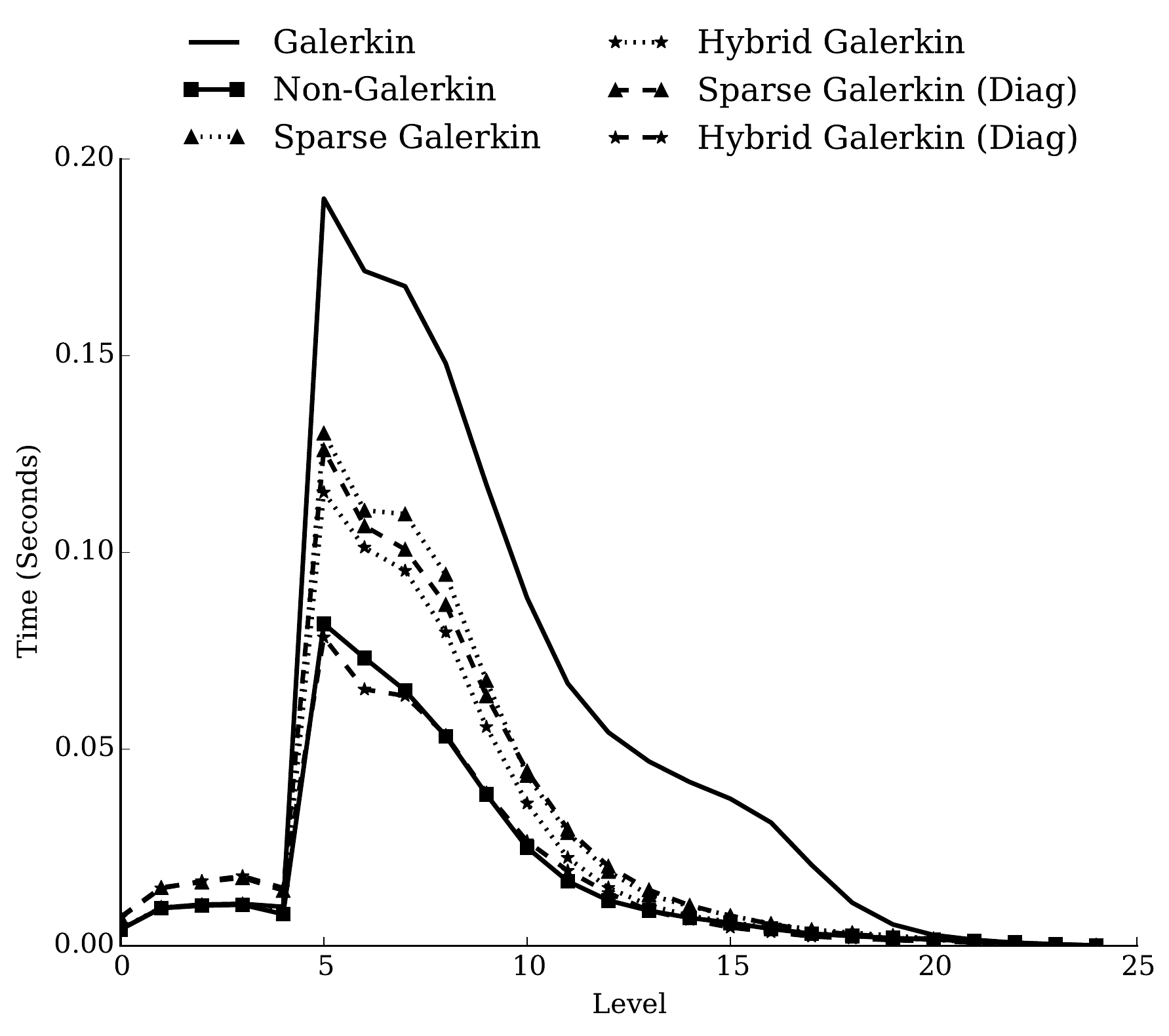}
	\caption{Time spent on each level of the AMG hierarchy during a single
	iteration of the solve phase for {\bf Laplace (left)} and {\bf rotated
	anisotropic diffusion (right)}, each with 10,000 degrees-of-freedom per
	core.}\label{figure:level_times_512_large}
\end{figure}
\begin{figure}[ht!]
	\centering
\includegraphics[width=0.49\textwidth]{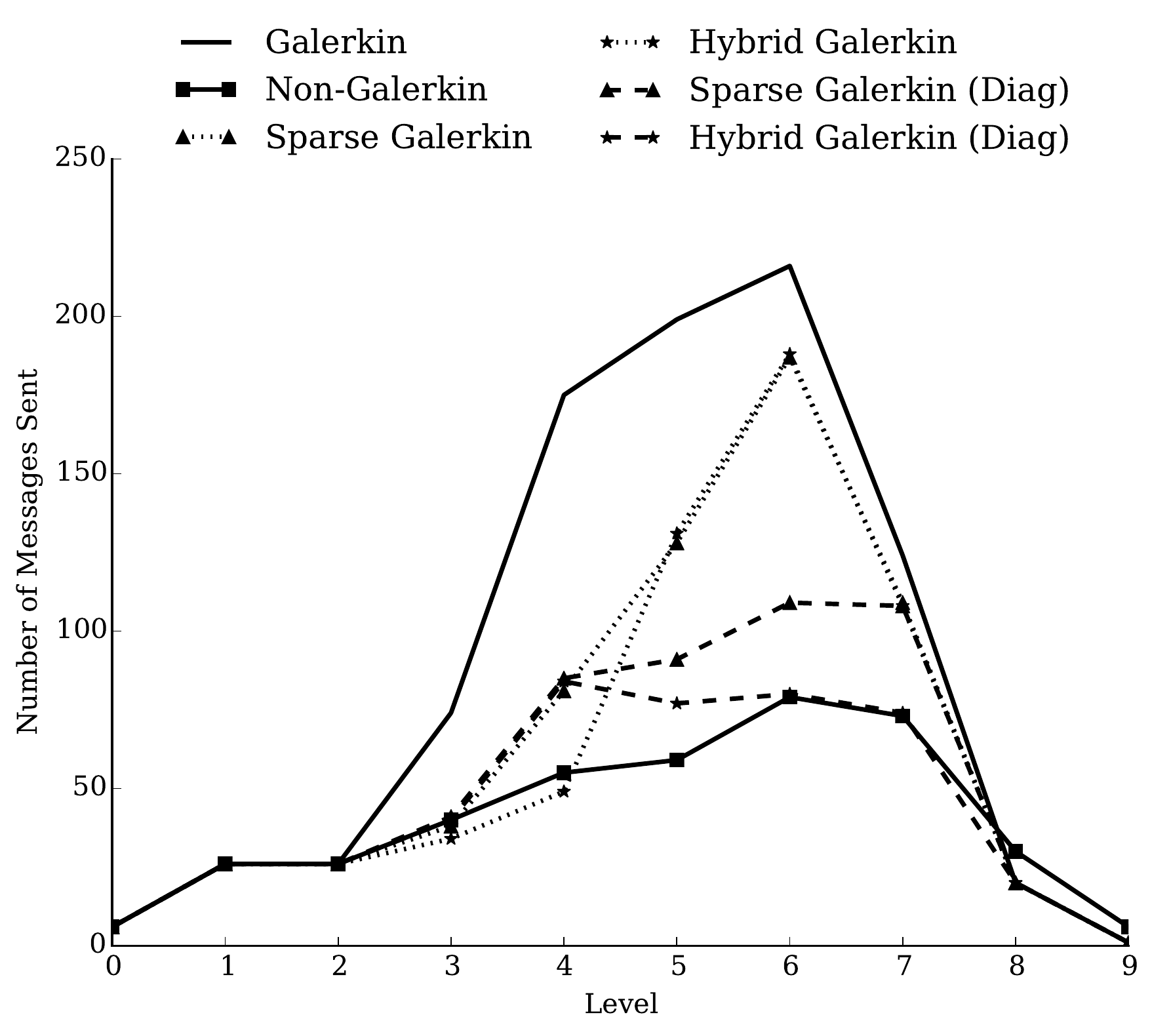}
\hfill
\includegraphics[width=0.49\textwidth]{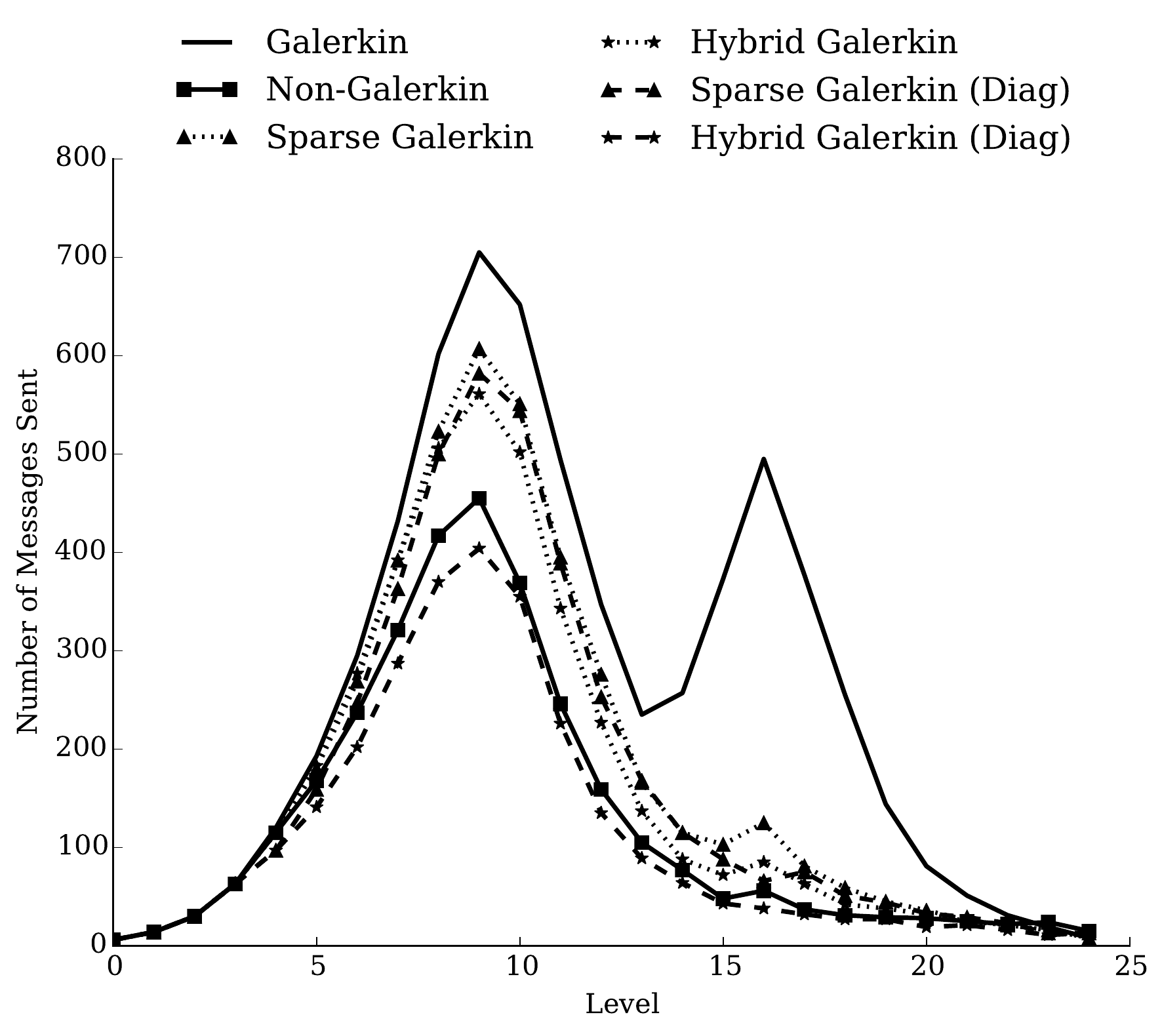}
\caption{Number of sends required to perform a single SpMV on each level of the
AMG hierarchy for: {\bf Laplace (left)} and {\bf Rotated anisotropic diffusion
(right)}, each with 10,000 degrees-of-freedom per core.}\label{figure:comm_512_large}
\end{figure}

The increase in time spent on each level, as well as the associated
communication costs of these levels, becomes more pronounced at higher processor
counts in a strong scaling study. Figure~\ref{figure:performance_512_small}
illustrates this by plotting the per-level times required
during a single iteration of AMG, as well as the number of messages communicated
during a SpMV for the rotated anisotropic diffusion problem with 1,250
degrees-of-freedom per core using 8192 cores.  Compared with the 10,000 degrees-of-freedom
per core example in Figures~\ref{figure:comm_512_large}~and~\ref{figure:level_times_512_large},
there is a sharper increase in time
required for levels near the middle of the hierarchy due to the increasing dominance of
communication complexity.
\textit{note:} Strong scaling the Laplace problem results in similar
performance.

\begin{figure}[ht!]
	\centering
\includegraphics[width=0.49\textwidth]{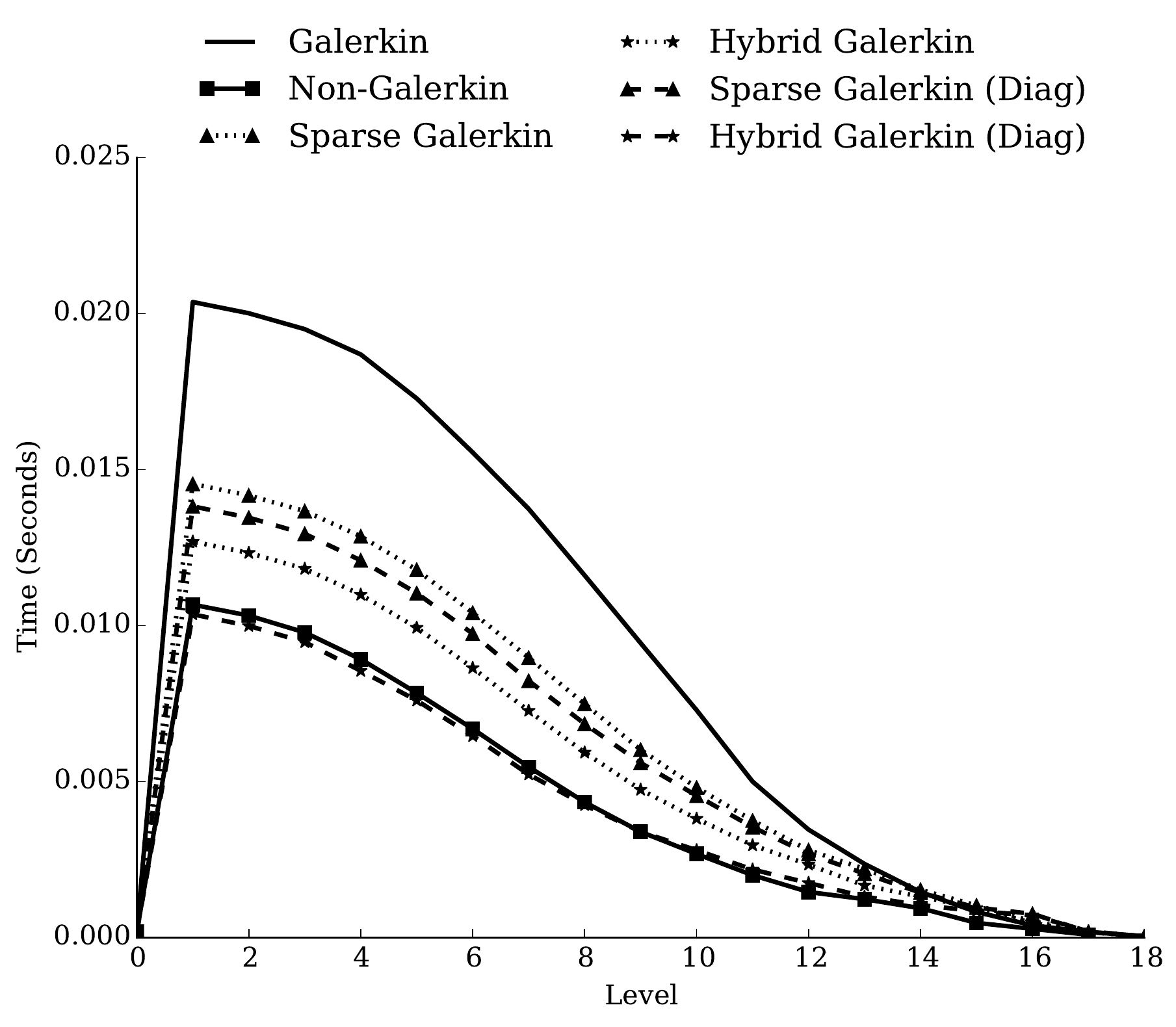}
\hfill
\includegraphics[width=0.49\textwidth]{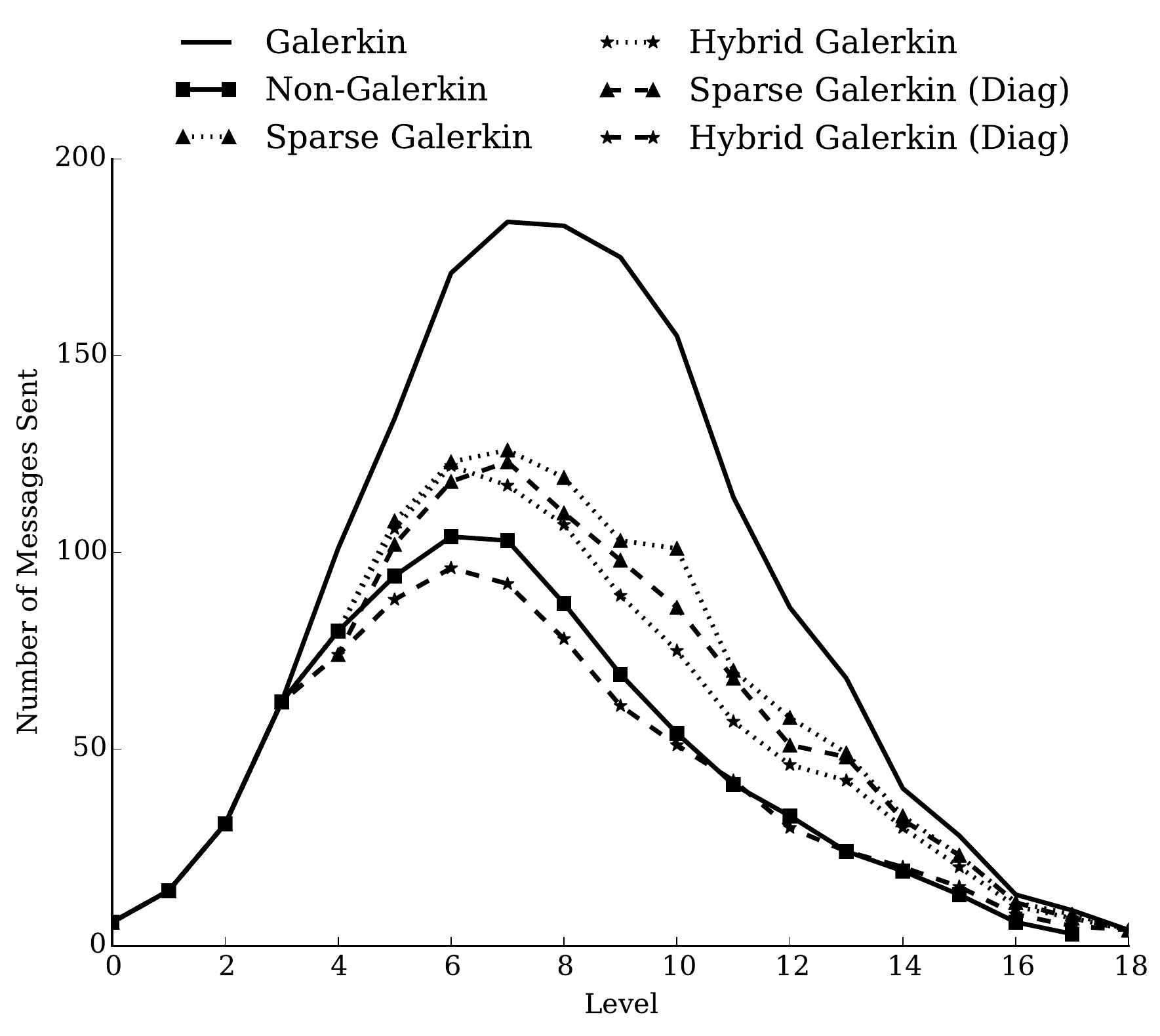}
	\caption{For each level of the AMG hierarchy, time per
	iteration of AMG (left) and number of messages sent during a
	single SpMV (right) for the {\bf rotated anisotropic diffusion} problem
	with 1,250 degrees-of-freedom per core.}\label{figure:performance_512_small}
\end{figure}

\subsection{Costs of Weakly Scaled Setup Phases}\label{section:results_setup}

Each sparsification method can lead to reduced communication costs in the
middle of the hierarchy. However, removing insignificant entries from
coarse-grid operators requires additional work in the setup phase. In the
non-Galerkin method, setup times are reduced since the increased sparsity
is used directly in the triple-matrix product required to form each
successive coarse-grid operator. However, for the new methods, Sparse and Hybrid Galerkin,
the entire Galerkin hierarchy is first constructed so that the sparsify process on
each level requires additional work.  Figure~\ref{figure:setup} shows the times
required to setup an AMG hierarchy for rotated anisotropic diffusion, with
Laplace setup times scaling in a similar manner.  While there is a slight
increase in setup cost associated with the Sparse and Hybrid Galerkin
hierarchies, this extra work is nominal.  Therefore, while the majority of this additional
 work is removed when using diagonal lumping, the differences in work required in the setup phase between these two lumping strategies is insignificant for the problems being tested.
\begin{figure}[ht!]
	\centering
\includegraphics[width=0.49\textwidth]{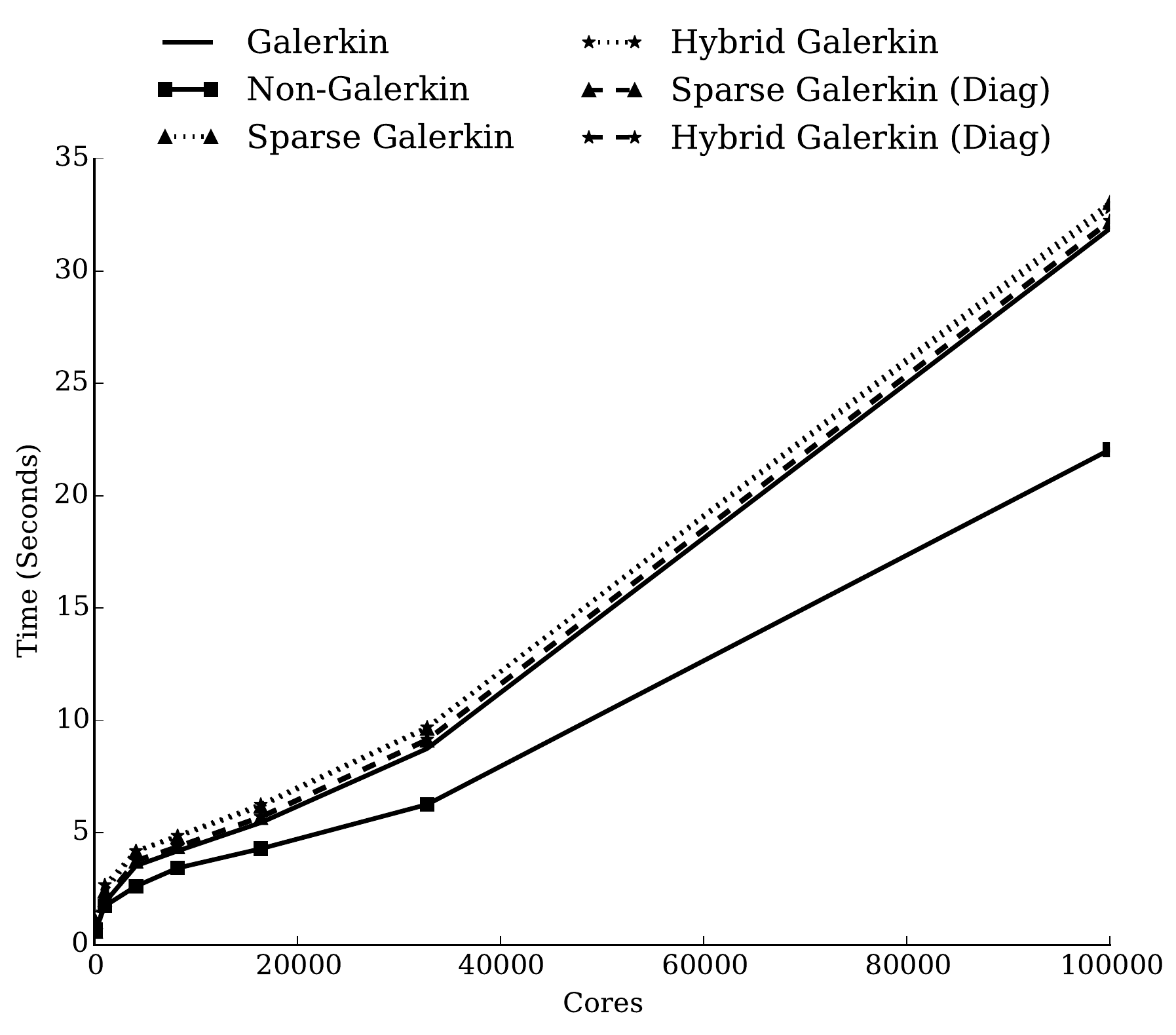}
	\caption{Time required to setup AMG hierarchy for {\bf rotated
	anisotropic diffusion} with 10,000 degrees-of-freedom per
	core.}\label{figure:setup}
\end{figure}

\subsection{Weak Scaling of GMRES Preconditioned by AMG}\label{section:results_gmres_weak}

In this section we investigate the \textit{weak} scaling properties of the
methods.  Figure~\ref{figure:aniso_times_scale_weak} shows both the average
convergence factor and total time spent in the solve phase for a weak scaling
study with rotated anisotropic diffusion problems at 10,000 degrees-of-freedom
per core using GMRES preconditioned by AMG\@. GMRES is used over CG because 
Algorithm~\ref{alg:sparsify} guarantees symmetry but not positive-definiteness
of the preconditioner.
In many cases, positive-definiteness is preserved, but when using more aggressive
drop tolerances, we have observed this property being lost.
While the convergence of both diagonally-lumped Sparse and Hybrid Galerkin remain
similar to that of Galerkin, the non-Galerkin method converges more slowly.
Therefore, while non-Galerkin and diagonally-lumped Hybrid Galerkin yield
similar communication requirements, GMRES preconditioned by Hybrid Galerkin
performs significantly better as fewer iterations are required.
\begin{figure}[ht!]
	\centering
	\includegraphics[width=0.49\textwidth]{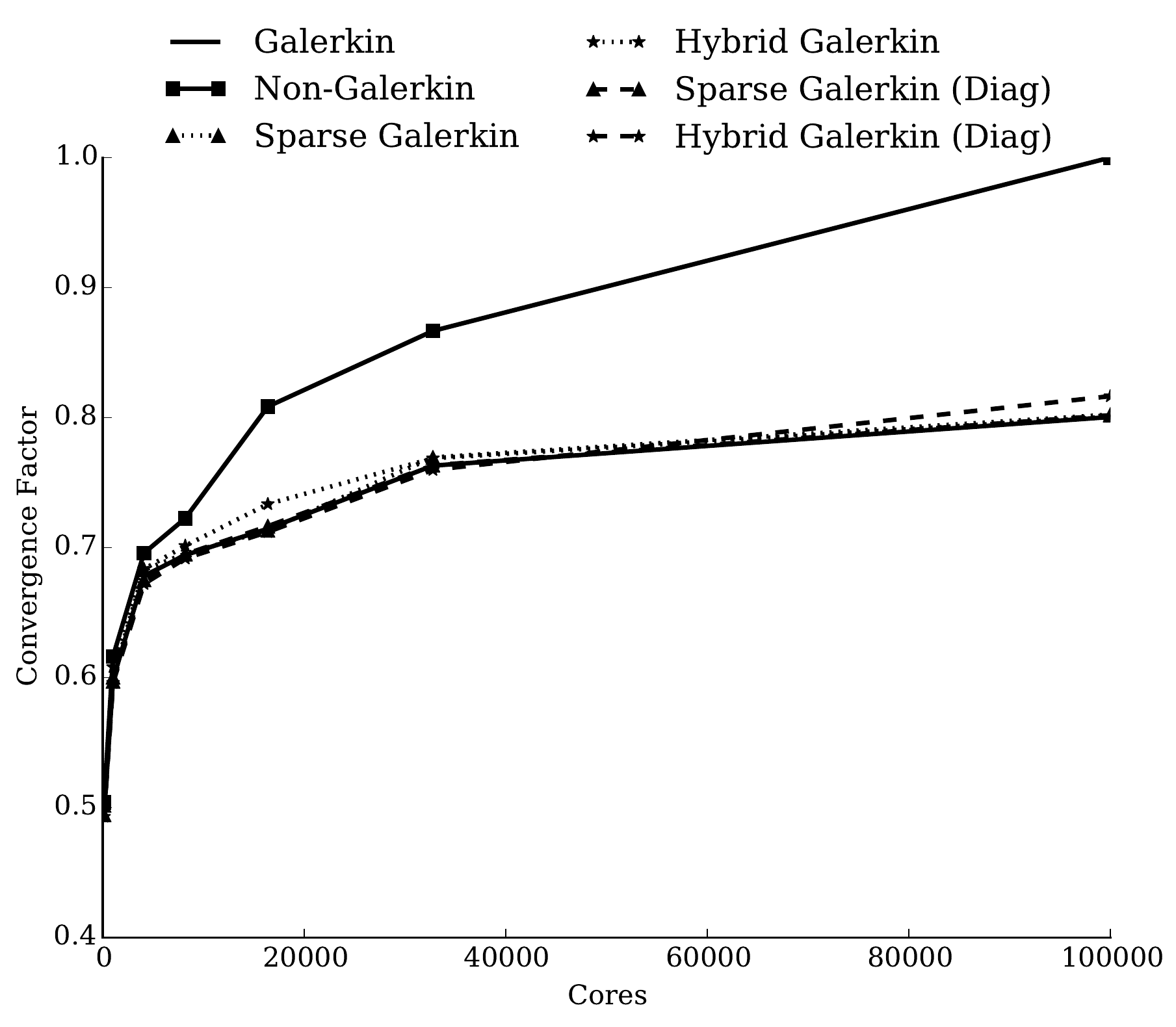}
  \hfill
\includegraphics[width=0.49\textwidth]{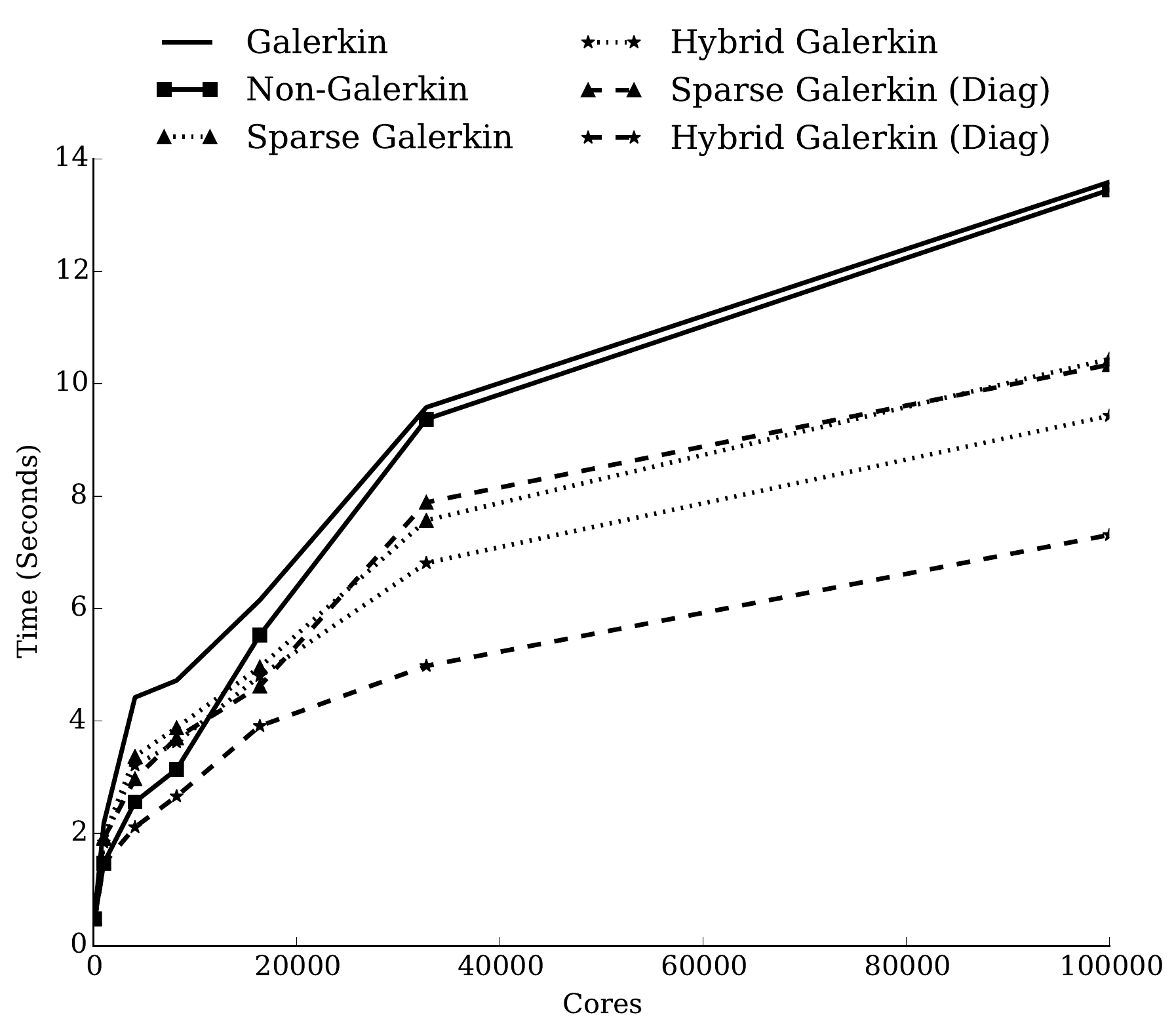}
	\caption{Convergence factors (left) and times (right)
	for weak scaling of {\bf anisotropic} problem (10,000 degrees-of-freedom
	per core), solved by preconditioned GMRES\@.
	For large problem sizes, non-Galerkin AMG does not converge, and
	timings indicate when the maximum iteration count was reached.}\label{figure:aniso_times_scale_weak}
\end{figure}
\begin{remark}
With the chosen drop tolerances, non-Galerkin does not converge for this
 anisotropic problem at $100,000$ cores.  In this case, nothing was dropped from
 the first three coarse-levels of the hierarchy.  On the fourth coarse-level a drop
tolerance of $0.01$ was used, and the fifth was sparsified with a tolerance of
$0.1$.  The remaining levels were sparsified with a drop tolerance of $1.0$.
This was determined to be the best tested drop tolerance sequence for smaller
run sizes, and multiple drop tolerance sequences were not tuned at this large
problem size due to the significant costs.  However, a better drop tolerance
 could yield a convergent non-Galerkin method at this scale.
\end{remark}

The efficiency of weakly scaling to $p$ processes is defined as $E_{p}
= \frac{T_{p}}{pT_{1}}$, where $T_{1}$ is the time required to solve the
problem on a single process and $T_{p}$ is time to solve on $p$ processes. The
efficiency of solving weakly scaled rotated anisotropic diffusion problems with
non-Galerkin, Sparse Galerkin, and Hybrid Galerkin, relative to the efficiency
of Galerkin AMG, are shown in Figure~\ref{figure:aniso_eff_scale_weak}. While
both the original and diagonally-lumped Sparse and Hybrid Galerkin methods
scale more efficiently than Galerkin, the poor convergence of non-Galerkin on large
run sizes yields a reduction in relative efficiency. While the methods perform
similarly when solving the Laplace problem, non-Galerkin improves relative
efficiency for all scalings of this model problem.
\begin{figure}[ht!]
	\centering
	\includegraphics[width=0.49\textwidth]{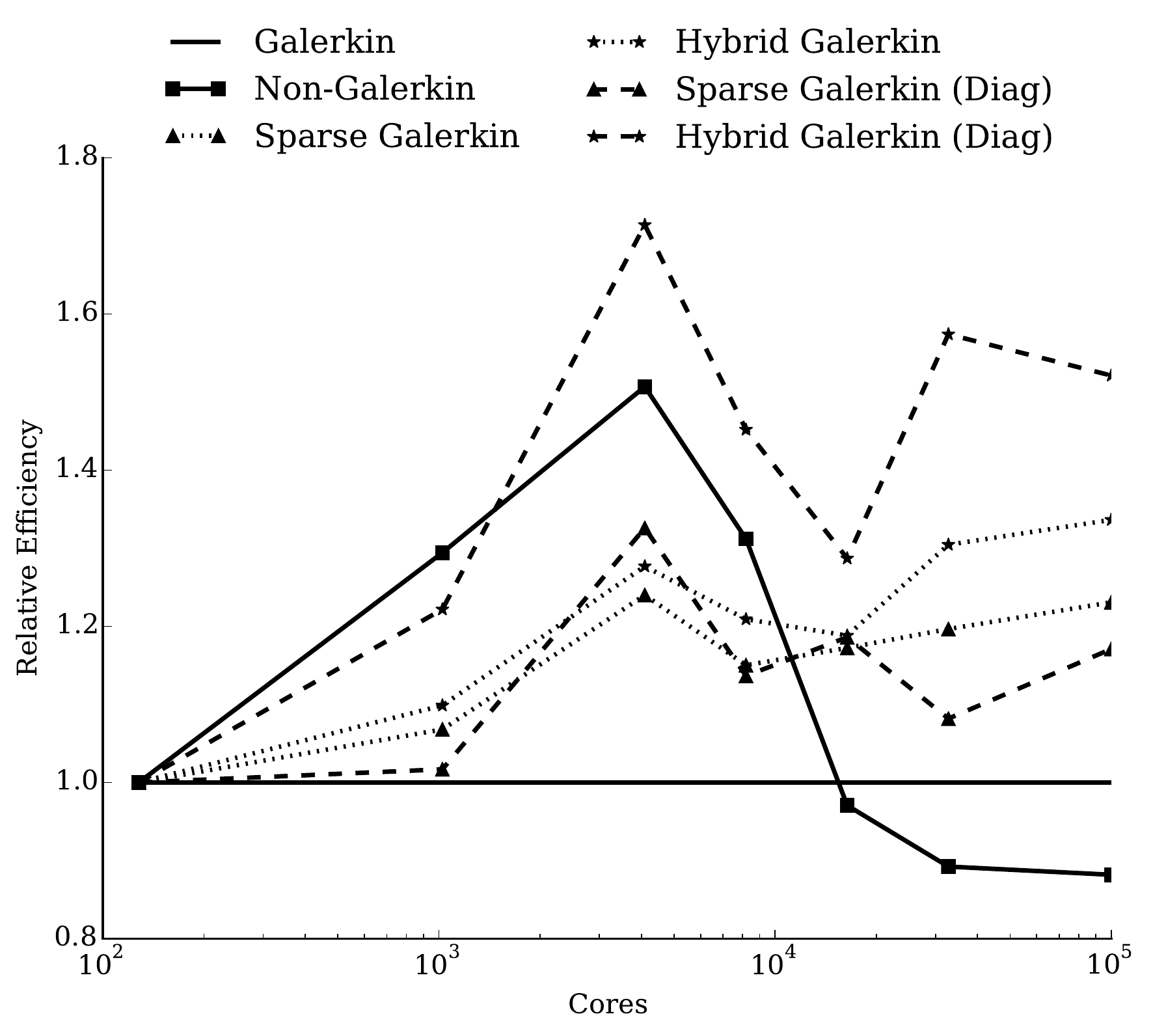}
	\caption{Efficiency of solving weakly scaled {\bf rotated anisotropic
	diffusion} at 10,000 degrees-of-freedom per core with various methods, relative
	to that of the Galerkin hierarchy.}\label{figure:aniso_eff_scale_weak}
\end{figure}

\subsection{Strong Scaling of GMRES Preconditioned by
AMG}\label{section:results_gmres_strong}

We next consider the rotated anisotropic diffusion system with approximately
$10,240,000$ unknowns using cores ranging from $128$ to $100,000$. Therefore,
the simulation is reduced from $80,000$ degrees-of-freedom per core when run on
$128$ cores, to just over $100$ degrees-of-freedom per core on $100,000$ cores.
Computation dominates the total cost of solving a problem partitioned over
relatively few processes, as each process has a large amount of local work.
However, as the problem is distributed across an increasing number of processes,
the local work decreases while communication requirements increase.  Therefore,
the time required to solve a problem is reduced with strong scaling, but only to
the point where communication complexity begins to dominate.  The efficiency of solving
this problem with preconditioned GMRES relative to Galerkin is shown in
Figure~\ref{figure:eff_scale_strong}.  In each case we observe improvements over
standard Galerkin.
\begin{figure}[ht!]
	\centering
	\includegraphics[width=0.49\textwidth]{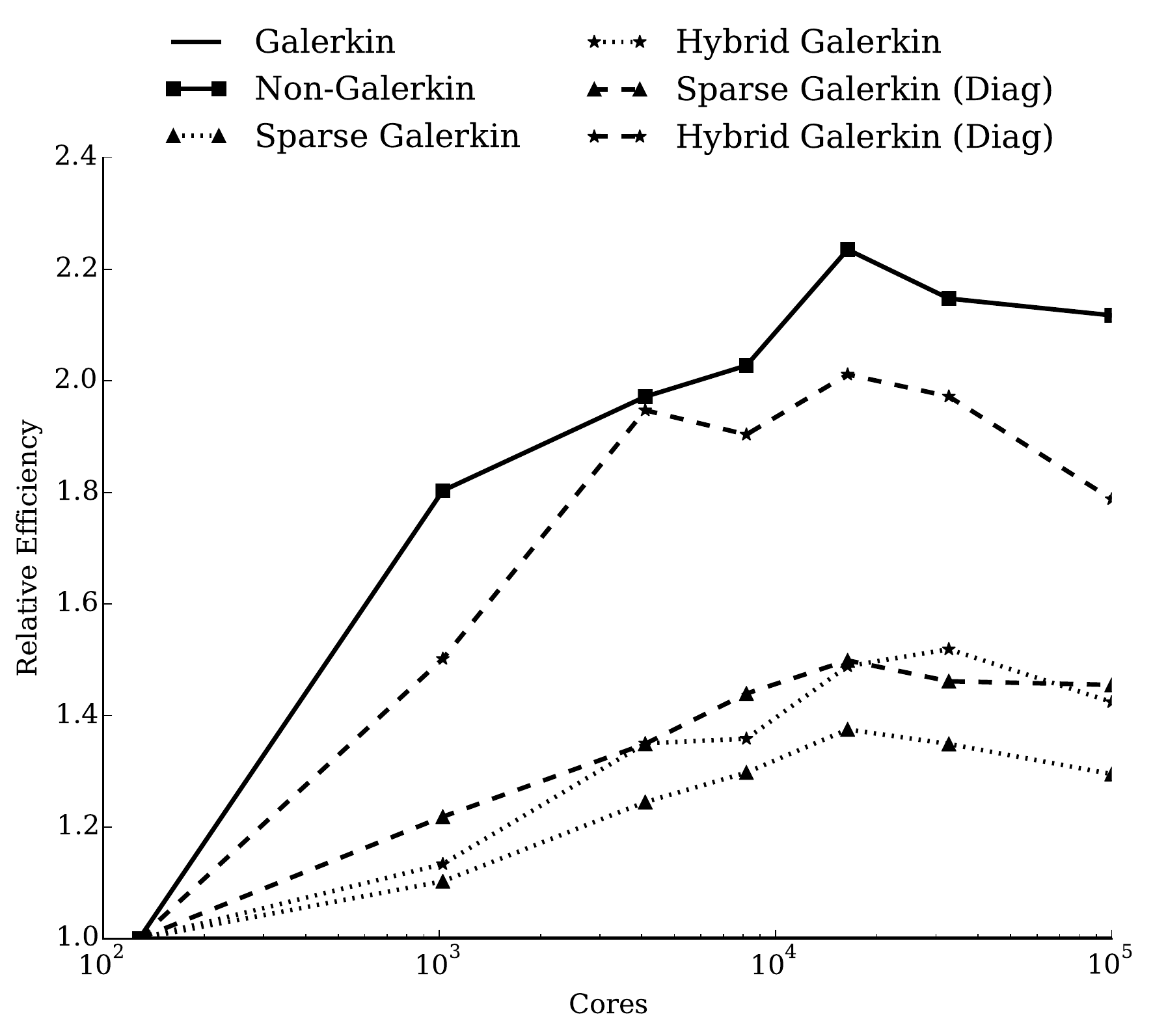}
	\caption{Efficiency of non-Galerkin and Sparse/Hybrid Galerkin methods in a
	strong scaling study, relative to Galerkin AMG for
	 {\bf rotated anisotropic diffusion}.}\label{figure:eff_scale_strong}
\end{figure}

A strong scaling study is also performed on the subset of matrices from the
Florida sparse matrix collection.  These problems were tested on $64$, $128$,
$256$, and $512$ processes.
Figure~\ref{figure:ufl_cycle_times} shows the time required to perform a single
V-cycle for each of the matrices in the subset, relative to the time required
by Galerkin AMG\@. All methods reduce the per-iteration times for each matrix
in the
subset. Furthermore, the total time required to solve each of these matrices
is also reduced, as shown in Figure~\ref{figure:ufl_times}. While Sparse
Galerkin provides some improvement, the Hybrid and non-Galerkin methods
are comparable, particularly at high core counts.
\begin{figure}[ht!]
	\centering
    \includegraphics[width=\textwidth]{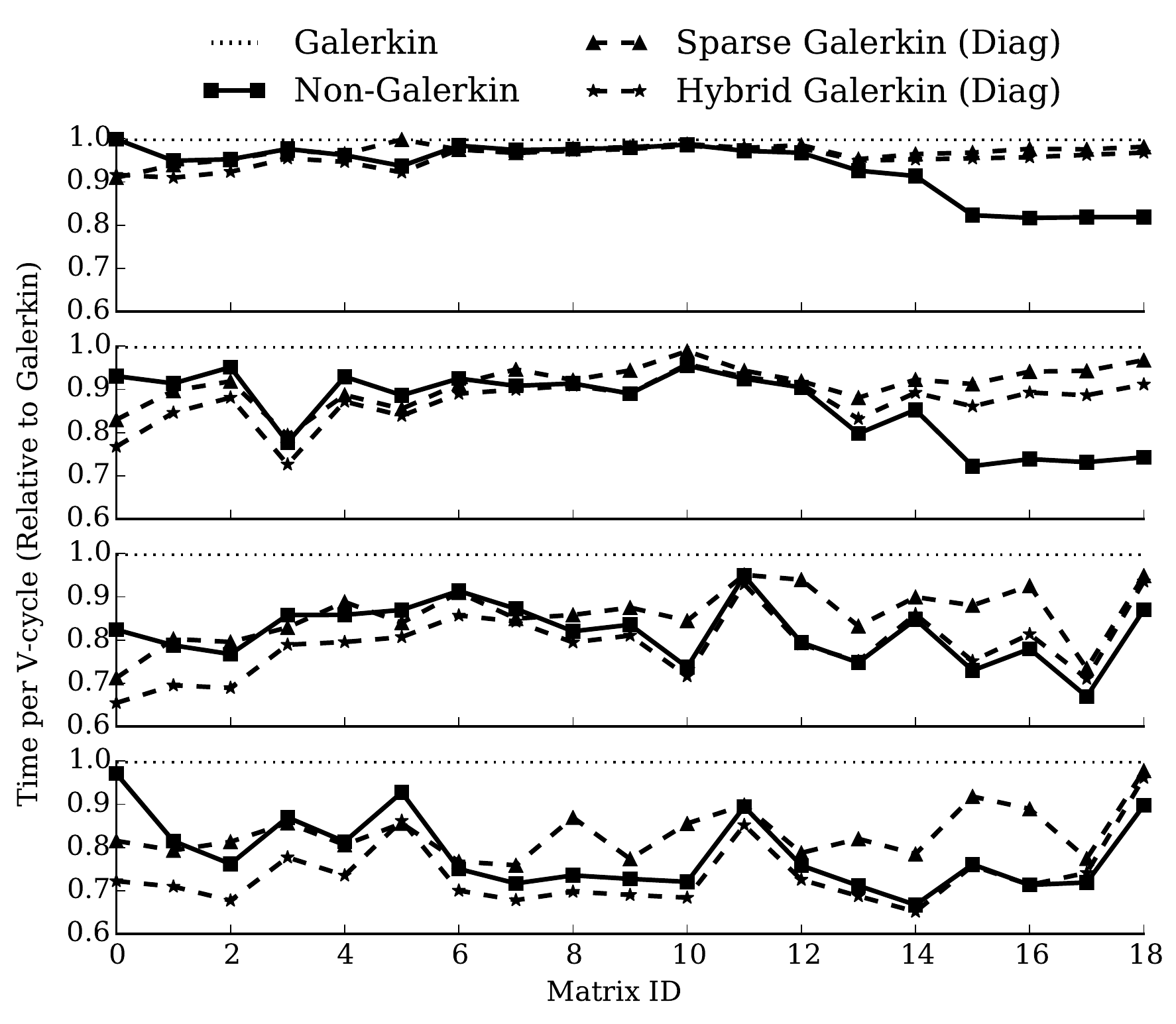}
\caption{Time (relative to Galerkin) per iteration for each matrix in the Florida Sparse Matrix Collection, using $p=64$, $128$, $256$, and $512$.}\label{figure:ufl_cycle_times}
\end{figure}
\begin{figure}[ht!]
	\centering
    \includegraphics[width=\textwidth]{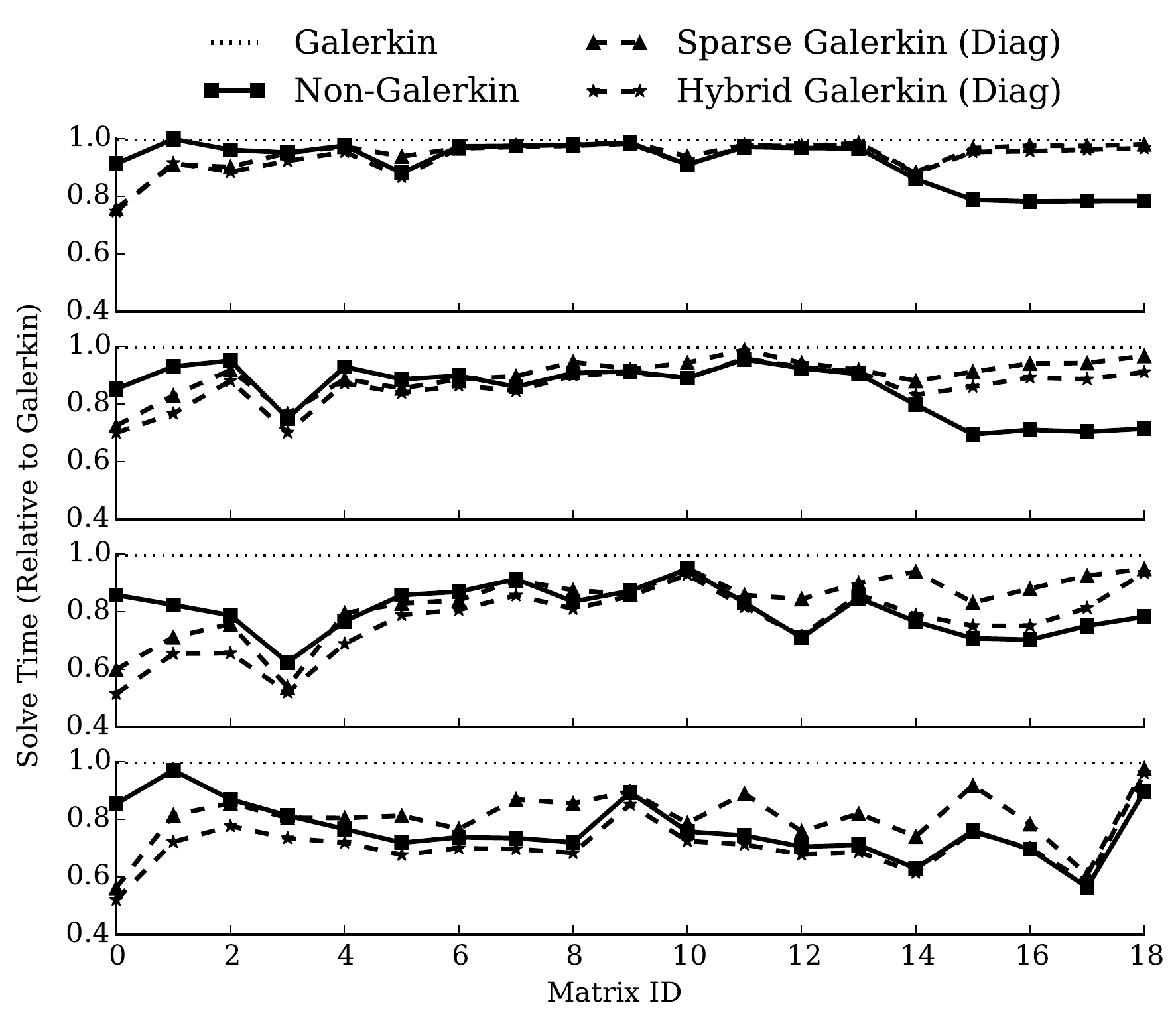}
  \caption{Time (relative to Galerkin) per AMG solve for each matrix in the Florida Sparse Matrix Collection, using $p=64$, $128$, $256$, and $512$.}\label{figure:ufl_times}
\end{figure}

\subsection{Diagonal Lumping Alternative and Preconditioned Conjugate Gradient}\label{section:results_pcg}

Diagonal lumping retains positive-definiteness of diagonally-dominant coarse-grid operators,
as described in Theorem~\ref{thm:spd}. Therefore, as the preconditioned
Conjugate Gradient (PCG) method requires both the matrix and preconditioner to
be symmetric and positive-definite, the Laplace and anisotropic diffusion
problems are solved by Conjugate Gradient preconditioned by the
diagonally-lumped Sparse and Hybrid Galerkin hierarchies.
Figure~\ref{figure:aniso_times_pcg} shows the solve phase times for solving the
weakly scaled rotated anisotropic diffusion problem with PCG\@.  As with GMRES,
both the Sparse and Hybrid Galerkin preconditioners decrease the time required
in the AMG solve phase during a weak scaling study.
\begin{figure}[ht!]
	\centering
	\includegraphics[width=0.49\textwidth]{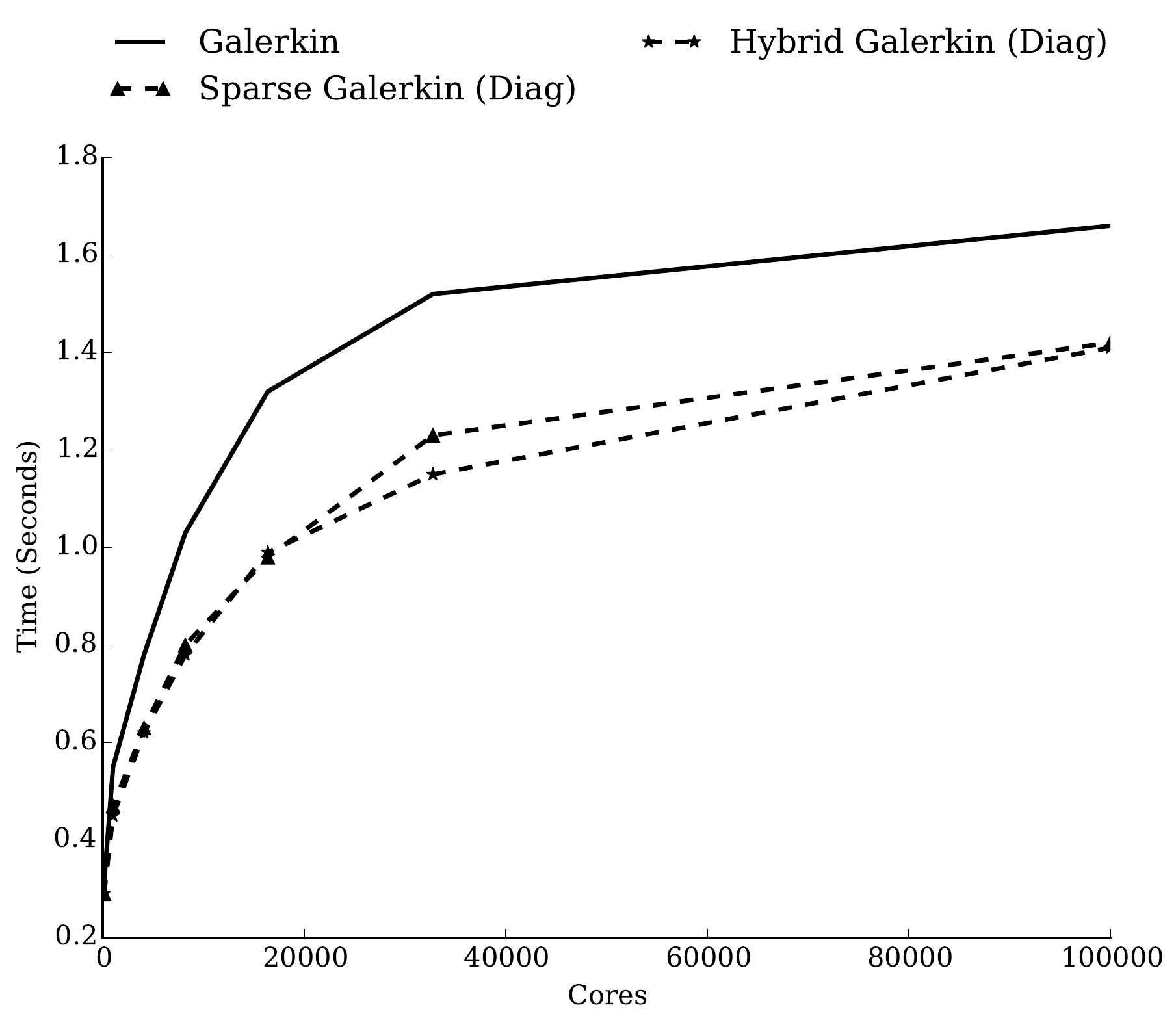}
	\caption{Weak scaling solve time for the {\bf anisotropic} problem, solved by PCG preconditioned by various AMG
	hierarchies.}\label{figure:aniso_times_pcg}
\end{figure}

\section{Adaptive Solve Phase}\label{section:adaptive}

The previous results describe the case of an optimal drop tolerance selected in
\code{sparsify}: the new diagonally-lumped Sparse and Hybrid Galerkin methods
reduce the cost of the solve phase.  However, as the optimal drop tolerance
changes with problem type, problem size, and even level of the AMG hierarchy, an
optimal drop tolerance is often not easily realized.  When the drop tolerance is
too small, few entries are removed from the hierarchy and the
communication complexity remains the same.  However, if the drop tolerance is too large, the
solver is non-convergent, as described in Section~\ref{section:nongal}.

In this section we consider an \textit{adaptive} method that attempts to add
entries back into the hierarchy as a deterioration in convergence is observed.
This is detailed in Algorithm~\ref{alg:adaptive}.  The algorithm initializes a
Sparse or Hybrid Galerkin hierarchy and proceeds by executing $k$ iterations of
a preconditioned Krylov method~---~e.g. PCG\@.  If the convergence is below a
tolerance, the coarse levels are traversed until a coarse grid operator is found
on which entries were removed with a drop tolerance greater than $0.0$.  Entries
are then added back to this coarse-grid operator, reducing the drop tolerance by
a factor of 10.  Any new drop tolerance below $\gamma_{\text{min}}=0.01$ is rounded down to $0.0$.
This continues until entries have been reintroduced into $s$ coarse-grid
operators.  At this point, the Krylov method continues, using the most recent
values for $x$ unless the previous iterations diverged from the true solution.
Many methods such as PCG and GMRES must be restarted after the preconditioner
has been edited.  This entire process is then repeated until convergence.  The
adaptive solve phase requires additional iterations over Galerkin AMG, as
initial iterations of this method may not converge. However, the goal of this
solver is to guarantee convergence similar to Galerkin AMG\@. Speed-up over
Galerkin AMG is still dependent on choosing reasonable initial drop tolerances.
\begin{algorithm}[!ht]
  \DontPrintSemicolon\KwIn{  \begin{tabular}[t]{l l}
    $A$, $b$, $x_{0}$\\
    $\shat{A}_{1}, \dots, \shat{A}_{\ell_{\max}}$ & Sparse/Hybrid Galerkin coarse grid matrices\\
    $A_{1}, \dots, A_{\ell_{\max}}$ & original Galerkin coarse grid matrices\\
    $P_{0}, \dots, P_{\ell_{\max}-1}$\\
    $k$ &  number of PCG iterations before convergence test\\
    $s$ & number of AMG levels to update at a time\\
    $\gamma_{0}$, $\gamma_{1}$, $\dots$ & drop tolerance used at each level for sparsification\\
    \code{tol} & convergence tolerance\\
	\code{sparse\_galerkin} & Sparse Galerkin method\\
	\code{hybrid\_galerkin} & Hybrid Galerkin method\\
  \end{tabular}
  }
	\KwOut{$x$}
  \;
	$x = x_{0}$\;
	$r_{0} = b - Ax_{0}$\;
  \;
	\While{$\sfrac{\|r\|}{\|r_{0}\|} \leq \code{tol}$}{    $M = \code{preconditioner}(\code{amg\_solve}, \shat{A}_{1}, \dots, \shat{A}_{\ell_{\max}}, P_{0}, \dots, P_{\ell_{\max}-1})$\;
		$x = \text{PCG}(A, b, x, k, M)$\tcc*[r]{Call $k$ steps of preconditioned CG}
		$r = b - Ax$\\
		\If{$\frac{\|r\|}{\|r_{0}\|} \leq \code{tol}$}
		{			\code{continue}\\
		}
		\Else		{      \For{$\ell=0,\dots,\ell_{\text{max}}$}{        \If{$\gamma_{0} > 0$}{          $\ell_{\text{start}} \leftarrow \ell$\tcc*[r]{Find finest level that uses dropping}
        }
      }
			\For{$\ell = \ell_{\text{start}} \ldots \ell_{\text{start}} + s$} {				$\gamma_{\ell} =
				\begin{cases}
                \frac{\gamma_{\ell}}{10}, & \text{if}\ \frac{\gamma_{\ell}}{10} > \gamma_{\text{min}} \\
      					0,                        & \text{\code{otherwise}}
        \end{cases}$\tcc*[r]{Determine new dropping parameter}
        \If(\tcc*[f]{Re-add entries at the new dropping tolerance}){\code{sparse\_galerkin}}{			$\shat{A}_{\ell} = $\code{sparsify}$(A_{\ell}$, $A_{\ell-1}$, $P_{\ell - 1}$, $S_{\ell-1}$, $\gamma_{\ell})$
		}
		\ElseIf(\tcc*[f]{Re-add entries at the new dropping tolerance}){\code{hybrid\_galerkin}}{			$\shat{A}_{\ell} = $\code{sparsify}$(A_{\ell}$, $\hat{A}_{\ell-1}$, $P_{\ell - 1}$, $S_{\ell-1}$, $\gamma_{\ell})$

		}
			}
		}
	}
	\caption{\code{adaptive\_solve}}\label{alg:adaptive}
\end{algorithm}

\begin{example}
As an example, consider the case of a hierarchy with 6 levels using drop
tolerances of $[0,0.01, 0.1, 1.0, 1.0, 1.0]$~---~i.e.,\ $\shat{A}_{1}$ retains
all entries from $A_{1}$, $ \shat{A}_{2}$ and $\shat{A}_{3}$ result from
\code{sparsify} with  $\gamma = 0.01$ and $\gamma = 0.1$, etc.  Suppose that
\code{adaptive\_solve} with $k = 3$ and $s = 2$ results in 3 iterations of PCG\@
and a large residual.  The adaptive solve find the first level containing a
sparsified coarse grid matrix, namely $\shat{A}_{2}$.  The drop tolerance on
this level is changed from $0.01$ to $0.0$, and the original coarse matrix
$A_{2}$ is sparsified with to the new drop tolerance.  Furthermore, since $s=2$
the drop tolerance on level $3$ is reduced from $0.1$ to $0.01$, and
$A_{3}$ is also sparsified.  PCG then
restarts with the new hierarchy.  If convergence continues to suffer after 3
iterations, the hierarchy is updated again, but since $\shat{A}_{2}$  has $\gamma_2 = 0.0$,
entries are reintroduced into coarse matrices $\shat{A}_{3}$ and $\shat{A}_{4}$ instead.
\end{example}

Using Algorithm~\ref{alg:adaptive}, Figure~\ref{figure:adaptive} shows
both the relative residual of the system after each iteration as well as the
communication costs of PCG using three different AMG hierarchies: standard
Galerkin, Sparse Galerkin with diagonal lumping and aggressive dropping, and
Sparse Galerkin with diagonal lumping modified with adaptivity.  
For the adaptive case, we
purposefully choose an overly aggressive initial drop tolerance so that entries can be added back
multiple times and one coarse level at a time to show the effect on convergence
and communication.  Initially, when the drop tolerance is aggressive, 
the associated communication costs are low,
but the resulting PCG iterations do not converge; this provides a
baseline. As sparse entries are reintroduced into
the hierarchy, convergence improves, while only slightly increasing the
associated communication cost.  When entries are reintroduced into the hierarchy, the
preconditioner for PCG changes, and hence, the method must be restarted.  After
restarting the method, convergence improves.
\begin{figure}[ht!]
	\centering
	\includegraphics[width=0.70\textwidth]{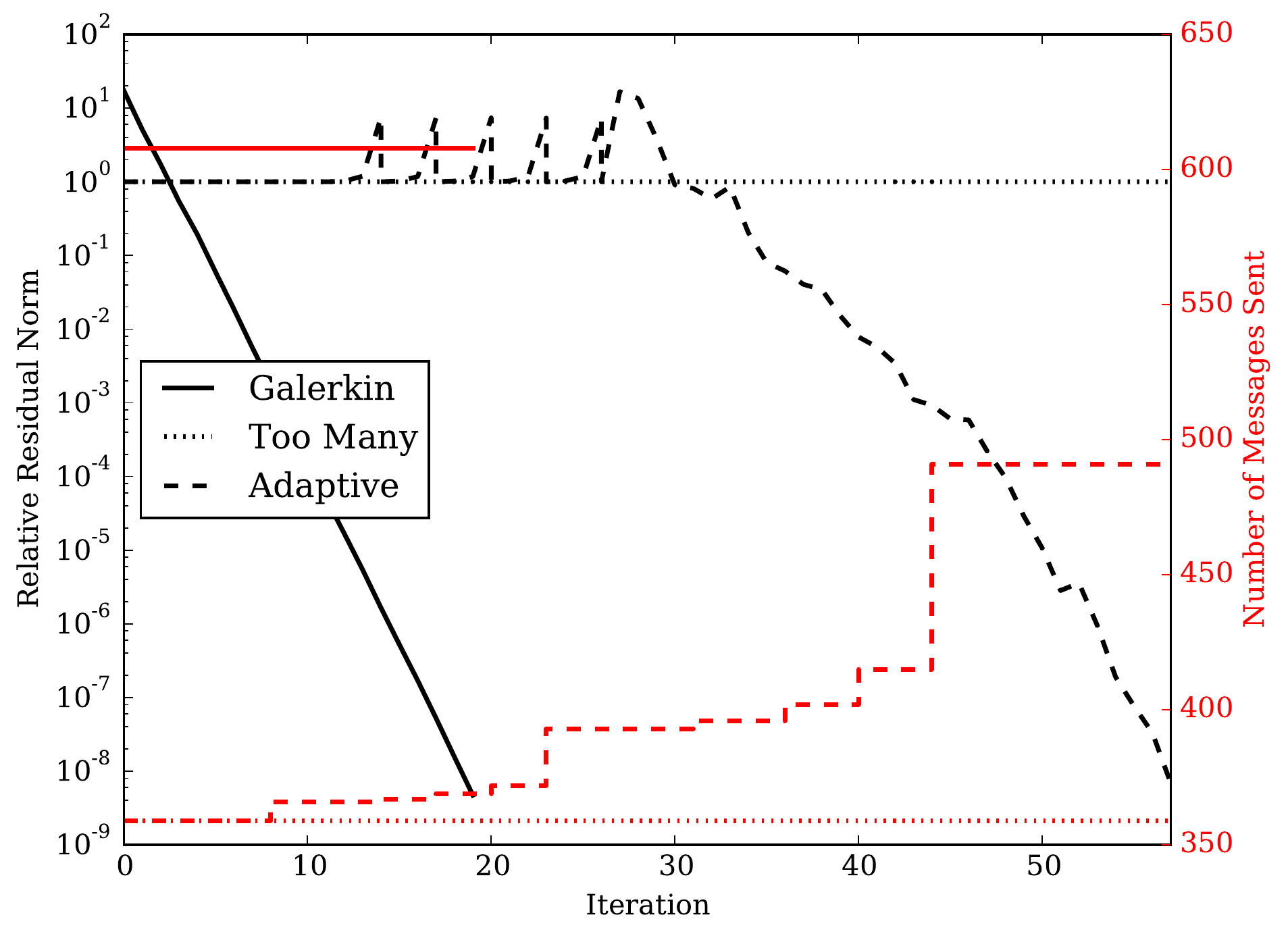}
\caption{Relative residual (black) and number of MPI sends (red) per
iteration when solving the {\bf Laplace} problem with: (1) PCG using Galerkin
AMG\@; (2) Hybrid Galerkin with aggressive dropping (labeled Too
Many); (3) Hybrid Galerkin solved with Algorithm~\ref{alg:adaptive},
using $k = 3$, $s = 1$, and $\gamma_{l} \dots \gamma_{l_{\max}}$ set to the same drop
  tolerances as the aggressive case.}\label{figure:adaptive}
\end{figure}

An important feature of Algorithm~\ref{alg:adaptive} is that it is solver
independent.  Indeed, other solvers may be used such as standalone AMG and
preconditioned FGMRES\@. No change to these solvers is needed when updating the
AMG hierarchy, as the change in hierarchy is considered during the solve.

\section{Conclusion}\label{section:conclusion}

We have introduced a lossless method to reduce the work required in parallel
algebraic multigrid by removing weak or unimportant entries from coarse-grid operators after the
multigrid hierarchy is formed.  This alternative to the original method of
non-Galerkin coarse-grids is similarly capable of reducing the communication costs on
coarse-levels, yielding an overall reduction in solve times.  Furthermore, this method
retains the original Galerkin hierarchy, allowing many of the restrictions of
non-Galerkin to be relaxed.  As a result, removed entries are easily lumped
directly to the diagonals, greatly reducing setup costs, while also
reducing communication complexity during the solve phase.  Furthermore,
as entries are added to the diagonal, entries removed from the matrix are stored
and adaptively reintroduced into the hierarchy if necessary for convergence.
Hence, the trade-off between convergence and the communication costs is controlled
at solve-time with little additional work.

\bibliographystyle{siam}
\bibliography{paper}

\begin{thebibliography}{10}

\bibitem{BlueWaters}
{\em {Blue Waters}}.
\newblock \url{https://bluewaters.ncsa.illinois.edu/}.

\bibitem{hypre}
{\em {HYPRE}: High performance preconditioners}.
\newblock \url{http://www.llnl.gov/CASC/hypre/}.

\bibitem{AsFa1996}
{\sc S.~F. Ashby and R.~D. Falgout}, {\em A parallel multigrid preconditioned
  conjugate gradient algorithm for groundwater flow simulations}, Nuclear
  Science and Engineering, 124 (1996), pp.~145--159.
\newblock UCRL-JC-122359.

\bibitem{bw-in-vetter13}
{\sc Brett Bode, Michelle Butler, Thom Dunning, Torsten Hoefler, William
  Kramer, William Gropp, and Wen mei Hwu}, {\em Contemporary High Performance
  Computing: From Petascale Toward Exascale}, vol.~1 of CRC Computational
  Science Series, Taylor and Francis, Boca Raton, 1~ed., 2013, ch.~The Blue
  Waters Super-System for Super-Science, pp.~339--366.

\bibitem{Bolten2015}
{\sc Matthias Bolten, Thomas~K. Huckle, and Christos~D. Kravvaritis}, {\em
  Sparse matrix approximations for multigrid methods}, Linear Algebra and its
  Applications,  (2015), pp.~--.

\bibitem{BrMcRu1984}
{\sc A.~Brandt, S.~F. McCormick, and J.~W. Ruge}, {\em Algebraic multigrid
  ({AMG}) for sparse matrix equations}, in Sparsity and Its Applications, D.~J.
  Evans, ed., Cambridge Univ. Press, Cambridge, 1984, pp.~257--284.

\bibitem{falgout}
{\sc E.~Chow, R.~D. Falgout, J.~J. Hu, R.~S. Tuminaro, and U.~M. Yang}, {\em A
  survey of parallelization techniques for multigrid solvers}, in Frontiers of
  Parallel Processing for Scientific Computing, Society for Industrial and
  Applied Mathematics, Philadelphia, PA, USA, 2005.

\bibitem{De1982}
{\sc JE~Dendy}, {\em Black box multigrid}, Journal of Computational Physics, 48
  (1982), pp.~366--386.

\bibitem{De1983}
{\sc J.~Dendy}, {\em Black box multigrid for nonsymmetric problems}, Appl.
  Math. Comput., 13 (1983), pp.~261--283.

\bibitem{ACCA}
{\sc Irad~Yavneh Eran~Treister, Ran~Zemach}, {\em Algebraic collocation coarse
  approximation (acca) multigrid}, in 12th Copper Mountain Conference on
  Iterative Methods, 2012.

\bibitem{NonGal_Schroder}
{\sc Robert~D. Falgout and Jacob~B. Schroder}, {\em Non-{G}alerkin coarse grids
  for algebraic multigrid}, SIAM Journal on Scientific Computing, 36 (2014),
  pp.~C309--C334.

\bibitem{Modeling}
{\sc Hormozd Gahvari, Allison~H. Baker, Martin Schulz, Ulrike~Meier Yang,
  Kirk~E. Jordan, and William Gropp}, {\em Modeling the performance of an
  algebraic multigrid cycle on hpc platforms}, in Proceedings of the
  International Conference on Supercomputing, ICS '11, New York, NY, USA, 2011,
  ACM, pp.~172--181.

\bibitem{ModelSpMV}
{\sc Hormozd Gahvari, Mark Hoemmen, James Demmel, and Katherine Yelick}, {\em
  Benchmarking sparse matrix-vector multiply in five minutes}, in {SPEC}
  Benchmark Workshop 2007, Austin, TX, January 2007.

\bibitem{BoomerAMG}
{\sc Van~Emden Henson and Ulrike~Meier Yang}, {\em Boomeramg: A parallel
  algebraic multigrid solver and preconditioner}, Appl. Numer. Math., 41
  (2002), pp.~155--177.

\bibitem{HPCC}
{\sc Piotr Luszczek, Jack~J. Dongarra, David Koester, Rolf Rabenseifner, Bob
  Lucas, Jeremy Kepner, John Mccalpin, David Bailey, and Daisuke Takahashi},
  {\em Introduction to the hpc challenge benchmark suite},  (2005).

\bibitem{McRu1982}
{\sc S.~F. McCormick and J.~W. Ruge}, {\em Multigrid methods for variational
  problems}, SIAM J. Numer. Anal., 19 (1982), pp.~924--929.

\bibitem{RuStu1987}
{\sc J.~W. Ruge and K.~St{\"{u}}ben}, {\em Algebraic multigrid ({AMG})}, in
  Multigrid Methods, S.~F. McCormick, ed., Frontiers Appl. Math., SIAM,
  Philadelphia, 1987, pp.~73--130.

\bibitem{DistTwo}
{\sc Hans~De Sterck, Robert~D. Falgout, Joshua~W. Nolting, and Ulrike~Meier
  Yang}, {\em Distance-two interpolation for parallel algebraic multigrid},
  Numerical Linear Algebra with Applications, 15 (2008), pp.~115--139.

\bibitem{AggCoarse2}
{\sc Hans~De Sterck, Ulrike~Meier Yang, and Jeffrey~J. Heys}, {\em Reducing
  complexity in parallel algebraic multigrid preconditioners}, SIAM J. Matrix
  Anal. Appl., 27 (2005), pp.~1019--1039.

\bibitem{NonGal_Treister}
{\sc Eran Treister and Irad Yavneh}, {\em Non-{G}alerkin multigrid based on
  sparsified smoothed aggregation}, SIAM Journal on Scientific Computing, 37
  (2015), pp.~A30--A54.

\bibitem{CCA}
{\sc Roman Wienands and Irad Yavneh}, {\em Collocation coarse approximation in
  multigrid}, SIAM Journal on Scientific Computing, 31 (2009), pp.~3643--3660.

\bibitem{AggCoarse}
{\sc Ulrike~Meier Yang}, {\em On long-range interpolation operators for
  aggressive coarsening}, Numerical Linear Algebra with Applications, 17
  (2010), pp.~453--472.

\end{thebibliography}

\end{document}